\documentclass[ba,preprint]{imsart}
\pubyear{202X}
\volume{TBA}
\issue{TBA}
\firstpage{1}
\lastpage{1}

\usepackage[utf8]{inputenc} 
\usepackage[english]{babel}

\usepackage{datetime}
\usepackage{multicol}
\usepackage{tcolorbox}

\usepackage{soul}

\newtcolorbox[auto counter]
  {block}
  {colback=gray!20,colframe=gray,fonttitle=\bfseries,
  title = \mytest{} procedure}
\makeatletter
\newcommand\tcb@cnt@blockautorefname{Procedure}
\makeatother

\newcommand\namedref[3][blue]{%
    \begingroup%
    \hypersetup{linkcolor=#1}%
    \hyperlink{#2}{#3}%
    \endgroup}

\usepackage{float}

\usepackage{caption, subcaption}
\usepackage{tikz}
\usepackage{xcolor}
\usepackage{graphicx}

\usepackage{amsfonts, amsmath, amssymb, amsthm, thmtools}
\usepackage{mathtools}
\usepackage{bm}
\usepackage{array,multirow,graphicx}
\usepackage{booktabs}
\usepackage{thmtools}
\usepackage{thm-restate}
\newtheorem{theorem}{Theorem}
\newtheorem{corollary}{Corollary}
\newtheorem{definition}{Definition}
\newtheorem{example}{Example}
\newtheorem{innercustomex}{Example}
\newenvironment{customex}[1]
  {\renewcommand\theinnercustomex{#1}\innercustomex}
  {\endinnercustomex}
\newtheorem{proposition}{Proposition}
\newenvironment{subproof}[1][Subproof]{%
  \begin{proof}[#1]%
}{%
  \end{proof}%
}
\DeclareMathOperator{\spn}{span}

\usepackage{nameref}
\usepackage{natbib}
\usepackage[colorlinks,citecolor=blue,urlcolor=blue,filecolor=blue,backref=page]{hyperref}
\usepackage{cleveref}

\usepackage{fontawesome5}
\usepackage{hyperref}

\newcommand{\TwoSample}{

\begin{tikzpicture}[x=0.75pt,y=0.75pt,yscale=-1,xscale=1]

\draw [color={rgb, 255:red, 0; green, 0; blue, 0 }  ,draw opacity=1 ][line width=1.5]  (13,199.23) -- (482,199.23)(52.25,25.25) -- (52.25,223) (475,194.23) -- (482,199.23) -- (475,204.23) (47.25,32.25) -- (52.25,25.25) -- (57.25,32.25)  ;
\draw [color={rgb, 255:red, 65; green, 117; blue, 5 }  ,draw opacity=1 ][line width=1.5]    (414.33,44.25) -- (364.32,65.66) -- (52.25,199.23) ;
\draw [color={rgb, 255:red, 186; green, 2; blue, 24 }  ,draw opacity=1 ][fill={rgb, 255:red, 208; green, 2; blue, 27 }  ,fill opacity=1 ][line width=1.5]    (139.83,55.25) -- (177.05,139.67) ;
\draw [shift={(178.67,143.33)}, rotate = 246.21] [fill={rgb, 255:red, 186; green, 2; blue, 24 }  ,fill opacity=1 ][line width=0.08]  [draw opacity=0] (11.61,-5.58) -- (0,0) -- (11.61,5.58) -- cycle    ;
\draw  [color={rgb, 255:red, 0; green, 0; blue, 0 }  ,draw opacity=1 ][fill={rgb, 255:red, 208; green, 2; blue, 27 }  ,fill opacity=1 ] (132.63,50.22) .. controls (132.63,47.68) and (134.68,45.63) .. (137.21,45.63) .. controls (139.74,45.63) and (141.79,47.68) .. (141.79,50.22) .. controls (141.79,52.76) and (139.74,54.81) .. (137.21,54.81) .. controls (134.68,54.81) and (132.63,52.76) .. (132.63,50.22) -- cycle ;

\draw (371.17,71.5) node  [color={rgb, 255:red, 65; green, 117; blue, 5 }  ,opacity=1 ,rotate=-336.51]  {$H_{0}: F_{X} = F_{Y}$};
\draw (173,85) node  [font=\footnotesize,rotate=-65.7]  {$\underset{( F_{X}, F_{Y}) \in H_{0}}{\text{\Large$\inf(d)$}}$};
\draw (484.5,193) node [anchor=north west][inner sep=0.75pt]   [align=left] {$\boldsymbol{F}_{X}$};
\draw (43,6) node [anchor=north west][inner sep=0.75pt]   [align=left] {$\boldsymbol{F}_{Y}$};
\draw (80,5) node [anchor=north west][inner sep=0.75pt]  [font=\huge] [align=left] {\underline{{\fontfamily{pcr}\selectfont Hypothesis Space:} $\mathbb{F}_X \times \mathbb{F}_Y$}};

\end{tikzpicture}
}

\newcommand{\mytest}{\texttt{PROTEST}} 

\startlocaldefs
\endlocaldefs

\begin{document}


\begin{frontmatter}
\title{\mytest{}: Nonparametric Testing of Hypotheses Enhanced by Experts' Utility Judgements}
\runtitle{\mytest{}}

\begin{aug}
\author{\fnms{Rodrigo F. L.} \snm{Lassance}\thanksref{addr1,t1,t2}\ead[label=e1]{rflassance@gmail.com}},
\author{\fnms{Rafael} \snm{Izbicki}\thanksref{addr1}\ead[label=e3]{rizbicki@gmail.com}}
\and
\author{\fnms{Rafael B.} \snm{Stern}\thanksref{addr2}\ead[label=e2]{rbstern@gmail.com}}

\runauthor{R.F.L. Lassance, R. Izbicki and R. B. Stern}

\address[addr1]{Department of Statistics, Federal University of São Carlos, São Paulo, Brazil.
}

\address[addr2]{Institute of Mathematics and Statistics, University of São Paulo, São Paulo, Brazil.
}

\thankstext{t1}{Institute of Mathematics and Computer Sciences, University of São Paulo, São Paulo, Brazil, \printead{e1}}

\end{aug}

\begin{abstract}
Instead of testing solely a precise hypothesis,
it is often useful to enlarge it with alternatives that are deemed to differ from it negligibly.
For instance, in a bioequivalence study
one might consider the hypothesis that
the concentration of an ingredient is 
exactly the same in two drugs.
In such a context, it might be more relevant to
test the enlarged hypothesis that
the difference in concentration between the drugs is
of no practical significance. While this concept is not alien to Bayesian statistics, applications remain confined to parametric settings and strategies on how to effectively harness experts' intuitions are often scarce or nonexistent. To resolve both issues, we introduce \mytest{}, an accessible nonparametric testing framework that seamlessly integrates with Markov Chain Monte Carlo (MCMC) methods. We develop expanded versions of the model adherence, goodness-of-fit, quantile and two-sample tests. To demonstrate how \mytest{} operates, we make use of examples, simulated studies \textendash{} such as testing link functions in a binary regression setting, as well as a comparison between the performance of \mytest{} and the PTtest \citep{holmes2015} \textendash{} and an application with data on neuron spikes. Furthermore, we address the crucial issue of selecting the threshold \textendash{} which controls how much a hypothesis is to be expanded \textendash{} even when intuitions are limited or challenging to quantify.
\end{abstract}


\end{frontmatter}

\section{Introduction}
\label{sct:intro}
Throughout the history of Bayesian statistics, the idea of inserting utility judgements directly into hypotheses has been often proposed, albeit remaining largely ignored in practical settings. The most pristine example of this behavior is perhaps the defense that all point null hypotheses should be reframed as composite ones \citep{edwards1963, good2009, berger1985}. However, this idea was either applied in very specific settings \textendash{} such as switching $H_0: \theta = \theta_0$ for $H_0: |\theta - \theta_0| \in [\delta_L, \delta_U]$, with $\delta_L \le 0 \le \delta_U$ known beforehand \citep{hobbs2007, kruschke2018} \textendash{} or not applied at all, being described as ``a lot of hard work'' \citep{leamer1988}.

The appeal of using external information to enlarge hypotheses is twofold, of both theoretical and practical nature. For the former, it avoids the requirement of adding probability masses to priors -- a common strategy when using Bayes factors \citep{jeffreys1961, kass1993, migon2014}. As for the latter, it allows for the inclusion of objective and subjective knowledge, such as measurement errors and researcher considerations on negligible deviations respectively, ensuring that the new hypothesis is more akin to the actual interest of the researcher.

This work brings forth a theoretical framework for hypothesis enlargement that is both capable of expanding nonparametric hypotheses based on the inputs of experts and easily applicable through currently available technologies, such as Markov Chain Monte Carlo (MCMC) methods. With this contribution, we expect researchers to be able to test complex hypotheses without having to disregard valuable information in the process. To achieve this end, we make use of pragmatic hypotheses \citep{hodges1954,esteves2019}. We propose a wider definition, contemplating cases that go beyond the parametric setting that was assumed in previous works:
\begin{definition}[Pragmatic hypothesis]
    \label{def:prag_hyp}
    Let $\mathbb{H}$ be the hypothesis space and $H_0 \subset \mathbb{H}$ be the null hypothesis of interest. For a given dissimilarity function $d(\cdot,\cdot)$ and a threshold $\varepsilon > \min(d) \ge 0$, a pragmatic hypothesis is defined as
    \begin{equation}
        Pg(H_0, d, \varepsilon) := \bigcup_{P_0 \in H_0}\{P \in \mathbb{H}: d(P_0,P) < \varepsilon\} = \left\{P \in \mathbb{H}: \inf_{P_0 \in H_0} d(P_0,P) <\varepsilon \right\}.
        \label{eq:prag_hyp}
    \end{equation}
    For brevity, if $d(\cdot, \cdot)$ and $\varepsilon$ are evident, we substitute $Pg(H_0, d, \varepsilon)$ for $Pg(H_0)$.
\end{definition}
\noindent In this paper, we usually set $\mathbb{H} = \mathbb{F}$, where $\mathbb{F}$ is the space of all distribution functions.

The intuition behind \autoref{def:prag_hyp} is as follows. The purpose of the pragmatic hypothesis is to expand the null so that it contains all elements that, for all practical purposes, are similar enough to at least one element of $H_0$. This is the same as checking, for each $P_0 \in H_0$, which are the elements $P \in \mathbb{H}$ such that $d(P_0, P) < \varepsilon$ to then take their union, which is represented by the left side of \eqref{eq:prag_hyp}. This is the same as evaluating, for each $P \in \mathbb{H}$, if the smallest difference between $P$ and all elements of $H_0$ is less than $\varepsilon$, the right side of \eqref{eq:prag_hyp}.

We provide three major contributions to the identification and use of pragmatic hypotheses in practical settings:
\begin{enumerate}
    \item Propose an intuitive testing procedure that can be seamlessly combined with MCMC methods (\mytest{}, \autoref{sct:over});
    \item Expand the theory of pragmatic hypotheses to nonparametric settings and explore how some hypotheses can be transformed into pragmatic ones (\autoref{sct:nph});
    \item Provide practical strategies for the choice of $\varepsilon$ even when it is not initially clear which value it should assume (\autoref{sct:epsilon}). This point is particularly important since defining $\varepsilon$ is often challenging.
\end{enumerate}
To ensure the adequacy of the procedure and demonstrate its applicability, we provide two simulated studies and an application with real data. The first simulated study (\autoref{sct:binary}) evaluates if \mytest{} can recover the true link function of binary data generated from a generalized linear model (GLM), while the second (\autoref{sct:NvsT}) is a comparison between \mytest{} and the PTtest \citep{holmes2015}. As for the application, it evaluates if data on neuron spikes resembles a Poisson process and if neurons behave differently between experiments (\autoref{sct:application}). Lastly, we discuss the potential of these methods and link it with other current research areas such as three-way testing (\autoref{sct:discussion}). The proofs of all results are presented in the appendix.

\begin{example}[Water droplet experiment]
\label{ex:drop1}
    The free falling water droplet experiment \citep{duguid1969} is a study that evaluates the behavior of small water droplets (ranging from 3 to 9 micrometers) as they fall through a tube in a controlled setting. One of the experiment's main objectives is to test the validity of Fick's law of diffusion, which in this case posits that the radius of the falling droplet changes linearly through time.
    
    As the droplet falls, a camera takes pictures of it every $0.5$ second and ceases recording after $7$ seconds. Therefore, $T = \{0s, 0.5s, \cdots, 6.5s, 7s\}$ represents the timestamps used as the independent variable. Consequently, two hypotheses of interest are
    $$
        \left\{\begin{array}{l}
            H_0^1: a(t) = \beta_0 + \beta_1 t, \quad \forall t \in T, \quad (\beta_0, \beta_1) \in \mathbb{R}^2;\\
            H_0^2: a(t) = \beta_0, \quad \forall t \in T, \quad \beta_0 \in \mathbb{R}.
        \end{array}\right.
    $$
    where $a(\cdot)$ represents the radius of the droplet at a given time. The first hypothesis represents Fick's law, while the second evaluates if time can be removed as a covariate.

    \begin{figure}[ht]
        \centering
        \includegraphics[width = \linewidth]{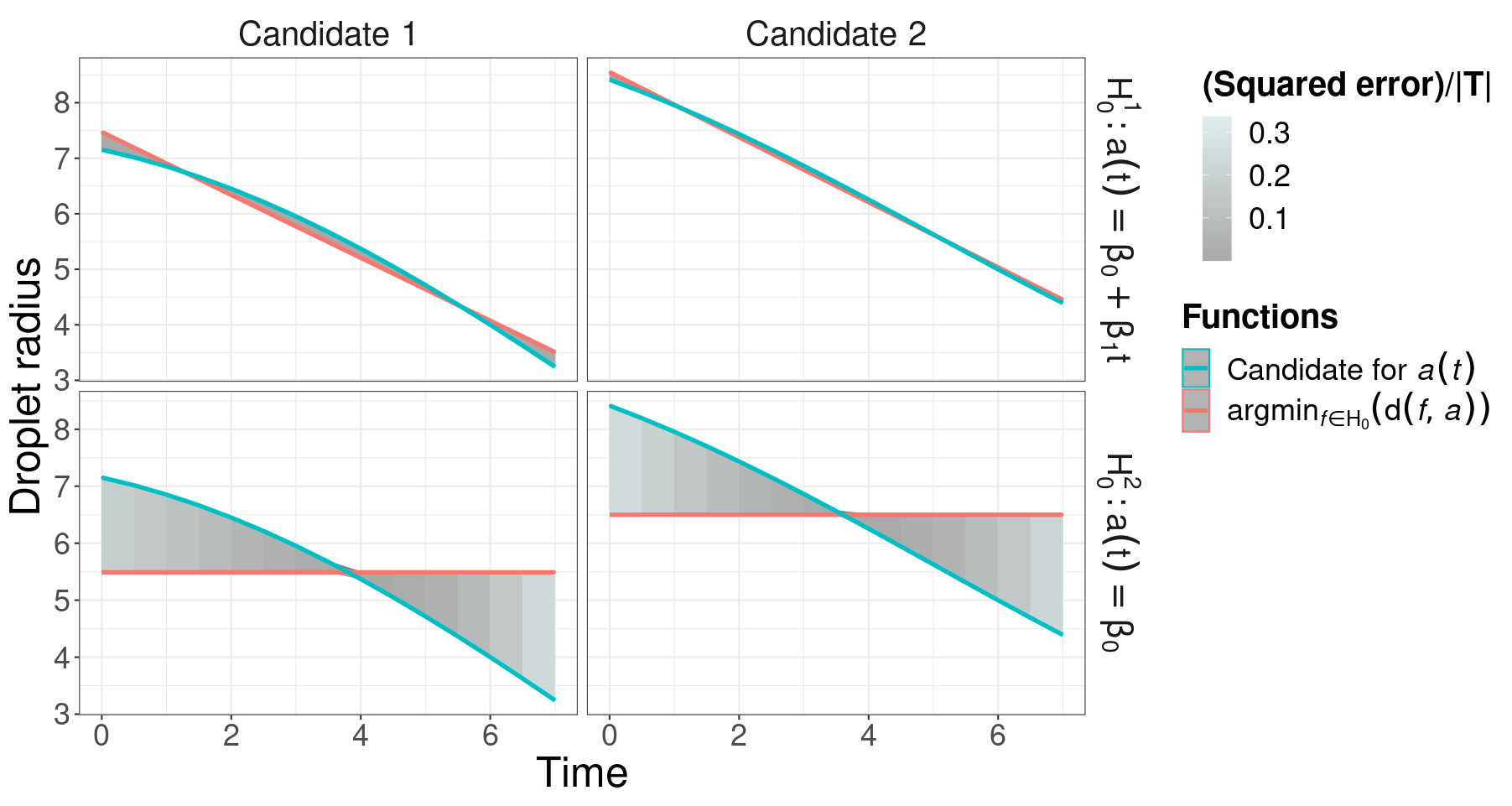}
        \caption{Candidates for $a(\cdot)$ (column) in \autoref{ex:drop1} and their best approximations under each hypothesis (row) based on the mean squared error between functions. The scale presents the point-wise squared error divided by $|T|$.}
        \label{fig:inf_rad}
    \end{figure}

    \autoref{fig:inf_rad} presents two viable functions for representing $a(\cdot)$ based on the data available (blue lines). Using as dissimilarity the square root of the expected squared error between two functions, we derive the linear functions that best approximate each under $H_0^1$ and $H_0^2$ (red lines). Following \autoref{def:prag_hyp}, the extent to which any of the linear functions is sufficiently similar to the original function depends on the dissimilarity being less than a threshold, which in this case is $\varepsilon \approx 0.1606$ (see \autoref{ex:drop13} for the reasoning behind this choice). For both cases, the dissimilarity falls under $\varepsilon$ on $H_0^1$ and over it on $H_0^2$, suggesting that Fick's law might be applicable for this case and that time should remain as a covariate.

\end{example}

\section{Overview}
\label{sct:over}
In this section, we define the Pragmatic Region Oriented TEST (\mytest{}, \autoref{def:test}) and provide an accessible guide for performing it (\namedref{blk:proc}{\mytest{} procedure}). The test is itself a variation of the Bayes decision for the 0-1-c loss function \citep{schervish2012} and directly evaluates the probability of $Pg(H_0)$.

\begin{definition}[Pragmatic region oriented test - \mytest{}]
    \label{def:test}
    Let $Pg(H_0, d, \varepsilon)$ be the pragmatic hypothesis, $\mathcal{P}$ be a random object over $\mathbb{H}$ and $\alpha \in [0,1]$. \mytest{} is such that
    \begin{itemize}
        \item If $\mathbb{P}\left(\mathcal{P} \in Pg(H_0)|\boldsymbol{X} = \boldsymbol{x}\right) \le \alpha$, reject the hypothesis;
        \item Otherwise, do not reject it.
    \end{itemize}
\end{definition}

From \autoref{def:prag_hyp}, we note that
\begin{equation}
    \mathbb{P}\left(\mathcal{P} \in Pg(H_0)|\boldsymbol{X} = \boldsymbol{x}\right) = \mathbb{P}\left(\inf_{P_0 \in H_0} d(P_0, \mathcal{P}) < \varepsilon \Big| \boldsymbol{X} = \boldsymbol{x}\right),
    \label{eq:prob_prag}
\end{equation}
which implies that the test can be conducted even when the full posterior is unknown or $Pg(H_0)$ cannot be fully specified. As long as $\inf_{P_0 \in H_0} d(P_0, P)$ can be obtained for every $P \in \mathbb{H}$, estimating \autoref{eq:prob_prag} becomes a matter of sampling from $\mathcal{P}|\boldsymbol{X} = \boldsymbol{x}$ and using the proportion of times in which $\inf_{P_0 \in H_0} d(P_0, \cdot) < \varepsilon$ as an estimate for \eqref{eq:prob_prag}. This is the motivation that leads to the \namedref{blk:proc}{\mytest{} procedure}, and ensures that it is fully compatible with MCMC methods and does not require knowledge of the full posterior distribution.

In the parametric setting, it is often possible to explicitly identify the pragmatic region since it is a subset of $\mathbb{R}^d$, so a posterior draw belongs to $Pg(H_0)$ if such subset contains it. This is not as straightforward when $\mathbb{H} = \mathbb{F}$, as $\mathcal{P}$ is a random object on the space of distribution functions. When dealing with hypotheses that reside in a function space, a more accessible strategy is to directly obtain $\inf_{P_0 \in H_0} d(P_0, P), \forall P \in \mathbb{H}$, which then allows for \autoref{eq:prob_prag} to be estimated through an MCMC sample.

\hypertarget{blk:proc}{}
\begin{block}
\begin{enumerate}
    \item Specify the null hypothesis $H_0$, the level $\alpha$, the dissimilarity function $d(\cdot, \cdot)$ and the threshold $\varepsilon$;
    \item Generate a sample $(\mathcal{P}^{(1)}, \mathcal{P}^{(2)}, \cdots, \mathcal{P}^{(N)})$, $N \in \mathbb{N}$, from the posterior distribution $\mathcal{P}|\boldsymbol{x}$;
    \item Obtain
    \begin{equation}
        \hat{\mathbb{P}}\left(\mathcal{P} \in Pg(H_0)|\boldsymbol{X} = \boldsymbol{x}\right) = \frac{1}{N}\sum_{i = 1}^{N} \mathbb{I}\left(\inf_{P_0 \in H_0} d(P_0, \mathcal{P}^{(i)}) < \varepsilon \right),
        \label{eq:prob_pg}
    \end{equation}
    where $\mathbb{I}(\cdot)$ is the indicator function;
    \item Reject the hypothesis if the estimated probability in \eqref{eq:prob_pg} is equal to or less than $\alpha$.
\end{enumerate}
\end{block}

\begin{customex}{1.1}[Water droplet experiment, continued]\label{ex:drop11}
Based on the \namedref{blk:proc}{\mytest{} procedure}, we set $\alpha = 0.05$ to finish the first step. For the second step, we apply a Gaussian process \citep{williams1996} with a Gaussian kernel to the data, using its posterior to draw regression functions for the test. For now, we omit how to achieve step 3 (see \autoref{ex:drop12} for the full discussion). The last step is evident from \autoref{fig:drop_test}. Based on the choice of $(\varepsilon, \alpha)$, we assert the validity of Fick's law and keep time as a covariate.

\begin{figure}[ht]
        \centering
        \includegraphics[width = \linewidth]{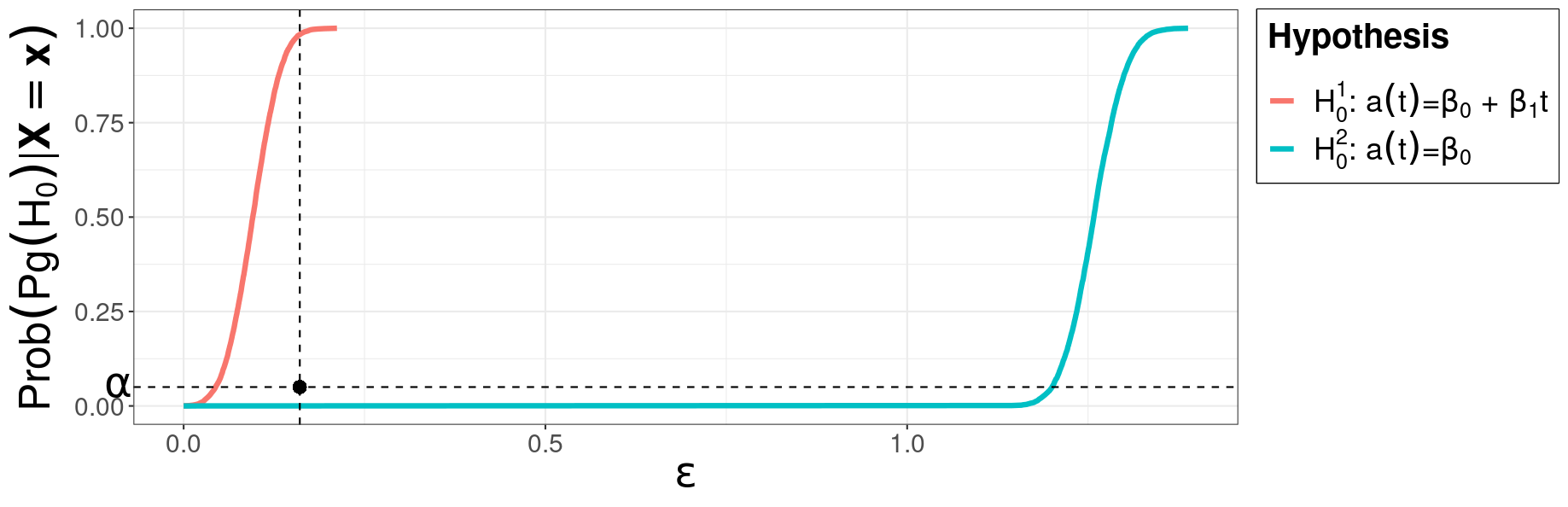}
        \caption{Largest $\varepsilon$ that entails rejection and the posterior probability of $Pg(H_0)$ for each hypothesis (colored curves) in \autoref{ex:drop11}. The black point represents the particular choice of $(\varepsilon \approx 0.1606, \alpha = 0.05)$, leading to not reject $H_0^1$ and reject $H_0^2$.}
        \label{fig:drop_test}
    \end{figure}
\end{customex}

In \autoref{ex:drop11}, we performed two tests in sequence -- $H_0^1$ and then $H_0^2$, with $H_0^2$ being more specific than $H_0^1$ -- while keeping $\alpha$ constant between tests. The following result ensures that \mytest{} cannot reach the counterintuitive conclusion of rejecting $H_0^1$ but not $H_0^2$ when $H_0^1 \supseteq H_0^2$ and $\alpha$ is fixed.

\begin{corollary}[Monotonicity property of \mytest{}]
\label{cor:monotonicity}
Let $H_0^1, H_0^2 \subset \mathbb{H}$ be such that
$$
    Pg(H_0^1, d, \varepsilon_1) \supseteq Pg(H_0^2, d, \varepsilon_2)
$$
and take $\alpha \in [0,1]$. If \mytest{} leads to rejecting $Pg(H_0^1)$, then it also rejects $Pg(H_0^2)$.
\end{corollary}

This is not usually the case for other standard testing procedures, such as the p-value \citep{schervish1996}. Moreover, even if the original two hypotheses are not nested, as long as their pragmatic versions are, this property will still hold for \mytest{}.

\section{Nonparametric pragmatic hypotheses}
\label{sct:nph}

In this section, we transform some common nonparametric hypotheses into pragmatic ones. From \autoref{def:prag_hyp}, this can be achieved by finding the infimum of the dissimilarity function between $H_0$ and any given $P \in \mathbb{H}$. We use the data to draw specific elements from $\mathbb{H}$ and the infimum to check if they belong to $Pg(H_0)$, rejecting the hypothesis if less than $\alpha\times 100$\% of them do. We can produce such elements by, for example, sampling from the Dirichlet or the Pólya tree processes \citep{ferguson1973, lavine1992, lavine1994}.

In some cases, the infimum can be obtained analytically and for a wide range of dissimilarity functions (such as in \autoref{sct:two_sample}), while in others it requires an optimization procedure (\autoref{sct:distribution}) or the choice of a specific dissimilarity (\autoref{sct:adherence}, \autoref{sct:quantile}). Whenever possible, the choice of the dissimilarity function should be based on how the researcher can best elicit their knowledge and interests about a problem. If that is not initially clear, we recommend the use of the classification dissimilarity due to its intuitive appeal.
\begin{definition}[Nonparametric classification dissimilarity function]
    \label{def:class_dissim}
    If $\mathbb{H} = \mathbb{F}$ and $F, G \in \mathbb{F}$ are distribution functions, the classification dissimilarity is given by
    \begin{equation}
        d_C(G, F) := 0.5 \left[\mathbb{P}\left(\frac{g(Z)}{f(Z)} > 1 \bigg| Z \sim G \right) + \mathbb{P}\left(\frac{f(Z)}{g(Z)} > 1 \bigg| Z \sim F \right)\right] \in [0.5,1],
        \label{eq:class_dissim}
    \end{equation}
    where $Z$ is a future observation, while $f$ and $g$ are the respective density functions of $F$ and $G$.
\end{definition}

The idea behind \autoref{eq:class_dissim} is as follows: say that there are two possible distribution functions ($F$ or $G$) that could be used to generate the future observation $Z$, and that there is no reason to assume one is more likely than the other, so $\mathbb{P}(Z \sim F) = \mathbb{P}(Z \sim G) = 0.5$. If the criteria for deciding from which distribution $Z$ came from is the likelihood ratio (LR), \eqref{eq:class_dissim} is the probability that the LR will favor the true distribution of the data. In other words, if $h_T(\cdot)$ is the true density function and $h_F(\cdot)$ is the other, then
\begin{align*}
    \mathbb{P}\left(\frac{h_T(Z)}{h_F(Z)} > 1 \right) &= 0.5 \times \mathbb{P}\left(\frac{h_T(Z)}{h_F(Z)} > 1 \bigg| Z \sim G \right) + 0.5 \times \mathbb{P}\left(\frac{h_T(Z)}{h_F(Z)} > 1 \bigg| Z \sim F \right)\\
    &= 0.5 \times \mathbb{P}\left(\frac{g(Z)}{f(Z)} > 1 \bigg| Z \sim G \right) + 0.5 \times \mathbb{P}\left(\frac{f(Z)}{g(Z)} > 1 \bigg| Z \sim F \right).
\end{align*}
Moreover, thanks to the Neyman-Pearson lemma \citep{neyman1933}, we conclude that the classification dissimilarity provides the highest achievable probability of correctly identifying which distribution function generated $Z$.

We present pragmatic versions of a model adherence test based on linear predictors (\autoref{sct:adherence}), the goodness-of-fit test (\autoref{sct:distribution}), the quantile test (\autoref{sct:quantile}) and the two-sample test (\autoref{sct:two_sample}). Whenever required, we apply a subscript to $\mathbb{F}$ to avoid ambiguity on what is the random variable being referenced. For example, if $X \in \mathbb{N}$, $\mathbb{F}_X$ may contain a Poisson distribution, but not a Normal distribution.

\subsection{Model adherence test}{\label{sct:adherence}}

We begin with a test focused on regression models applicable to data $(\boldsymbol{y},\boldsymbol{X})$, and therefore $\mathbb{H}$ is a space of functions of the type $g:  \mathcal{X} \longrightarrow \mathbb{R}$, where $\mathcal{X}$ is the covariates' domain. Our main finding (\autoref{thm:lm_inf}) shows how to analytically obtain the pragmatic hypothesis when comparing a function to a class of linear models.

\begin{theorem}[Linear model test]\label{thm:lm_inf}
    Let $\mathbb{H}$ be such that
    $$
        g \in \mathbb{H} \Longleftrightarrow \mathbb{E}_X(g^2) = \int_{\mathcal{X}} g(\boldsymbol{x})^2 d\mathbb{P}(\boldsymbol{x}) < \infty.
    $$
    Let $\boldsymbol{b}(\boldsymbol{x}) = (b_1(\boldsymbol{x}), b_2(\boldsymbol{x}), \cdots, b_k(\boldsymbol{x})) \subset \mathbb{H}$ be a linearly independent set of linear functions and choose $H_0$ such that
    \begin{equation}
        \label{eq:lin_hyp}
        H_0: R(\boldsymbol{x}) = \boldsymbol{b}(\boldsymbol{x})\boldsymbol{\beta}, \quad \forall \boldsymbol{x} \in \mathcal{X}, \quad \boldsymbol{\beta} \in \mathbb{R}^k,
    \end{equation}
    where $R(\cdot)$ is the true regression function. If $d(f, g) := \sqrt{\mathbb{E}_X[(f - g)^2]}$, then  $d(H_0, g) = d(\boldsymbol{b} \times \hat{\boldsymbol{\beta}}, g)$ for any $g \in \mathbb{H}$, where
    \begin{align*}
        \hat{\boldsymbol{\beta}} = A_{\boldsymbol{b}}^{-1} \times \boldsymbol{g}_{\boldsymbol{b}},
        \quad A_{\boldsymbol{b}} &=
        \left(\begin{array}{cccc}
             \mathbb{E}[b_1^2(\boldsymbol{X})] & \mathbb{E}[b_2(\boldsymbol{X}) b_1(\boldsymbol{X})] & \cdots & \mathbb{E}[b_k(\boldsymbol{X}) b_1(\boldsymbol{X})] \\
             \mathbb{E}[b_1(\boldsymbol{X}) b_2(\boldsymbol{X})] & \mathbb{E}[b_2^2(\boldsymbol{X})] & \cdots & \mathbb{E}[b_k(\boldsymbol{X}) b_2(\boldsymbol{X})] \\
             \vdots     & \vdots     & \ddots & \vdots     \\
             \mathbb{E}[b_1(\boldsymbol{X}) b_k(\boldsymbol{X})] & \mathbb{E}[b_2(\boldsymbol{X}) b_k(\boldsymbol{X})] & \cdots & \mathbb{E}[b_k^2(\boldsymbol{X})]
        \end{array}\right),\\
        \boldsymbol{g}_{\boldsymbol{b}}' &= \Big(\mathbb{E}[g(\boldsymbol{X}) b_1(\boldsymbol{X})], \quad \mathbb{E}[g(\boldsymbol{X}) b_2(\boldsymbol{X})], \quad \cdots \quad, \quad \mathbb{E}[g(\boldsymbol{X}) b_k(\boldsymbol{X})] \Big).
    \end{align*}
\end{theorem}

Some lingering aspects of \autoref{thm:lm_inf} require additional explanations. Framing the hypothesis as the span of linearly independent functions allows for testing a diverse set of assumptions, some of them being: $\boldsymbol{b}(\boldsymbol{x}) = \boldsymbol{x}$ (standard linear regression), $\boldsymbol{b}(\boldsymbol{x}) = \boldsymbol{x}_{-i}$ (removal of the $i$-th entry of the vector, thus providing a variable selection procedure) and $\boldsymbol{b}(\boldsymbol{x}) = (x_1+x_2, \boldsymbol{x}_{-\{1,2\}}')'$ (first two entries receive the same parameter $\beta$). As for the choice of the probability measure $\mathbb{P}$, if the context of the problem is not sufficient to imply one, we suggest using the empirical distribution of $\boldsymbol{b}(\boldsymbol{X})$, which leads to $\hat{\boldsymbol{\beta}} = (\boldsymbol{b}(\boldsymbol{X})'\boldsymbol{b}(\boldsymbol{X}))^{-1}\boldsymbol{b}(\boldsymbol{X})'g(\boldsymbol{X})$.

\begin{customex}{1.2}[Water droplet experiment, continued] \label{ex:drop12}
    As mentioned in \autoref{ex:drop1}, the covariate $T$ is a discrete variable and all times are recorded in the experiment, therefore it is reasonable to assign a discrete uniform distribution to it. Hence
    \begin{equation*}
        d(f,g) = \sqrt{\int_{\mathcal{X}} |f(\boldsymbol{x}) - g(\boldsymbol{x})|^2 d\mathbb{P}(\boldsymbol{x})} = \sqrt{\frac{1}{15}\sum_{t \in T} |f(t) - g(t)|^2},
    \end{equation*}
    which is a weighted version of the $l^2$ distance. From \autoref{thm:lm_inf} and assuming $T$ to be a column vector,  $\boldsymbol{b}(\boldsymbol{X}) = (\boldsymbol{1}, T)$ for $H_0^1$ and $\boldsymbol{b}(\boldsymbol{X}) = \boldsymbol{1}$ for $H_0^2$.

    Once again, we use the Gaussian process to model the data, this time applying different kernels (exponential, Gaussian and Matérn) for more robust results. Since $T$ is discrete, it is only necessary to obtain draws of the regression function at its values. The remaining steps of the \namedref{blk:proc}{\mytest{} procedure} lead us to apply \autoref{thm:lm_inf} to obtain the dissimilarity between each draw and $H_0$ and then take the proportion of times such dissimilarity is less than 0.1606 to reach a decision.

    \autoref{fig:drop-test} presents the results of both tests, leading to non-rejection for $H_0^1$ and to rejection for $H_0^2$. For $H_0^1$, \autoref{fig:drop1} shows that the significance level required to reject the hypothesis would be of at least 0.25, leading to the conclusion that the droplet radius can indeed be described as a linear function of the time. As for $H_0^2$, the threshold of choice leads to rejection in all cases, with a considerably higher value required for concluding otherwise. Hence, not only can the data be described by a linear model, but it also requires time to be kept as a covariate.

    \begin{figure}[ht]
         \centering
         \begin{subfigure}[b]{0.49\textwidth}
             \centering
             \includegraphics[width=\textwidth]{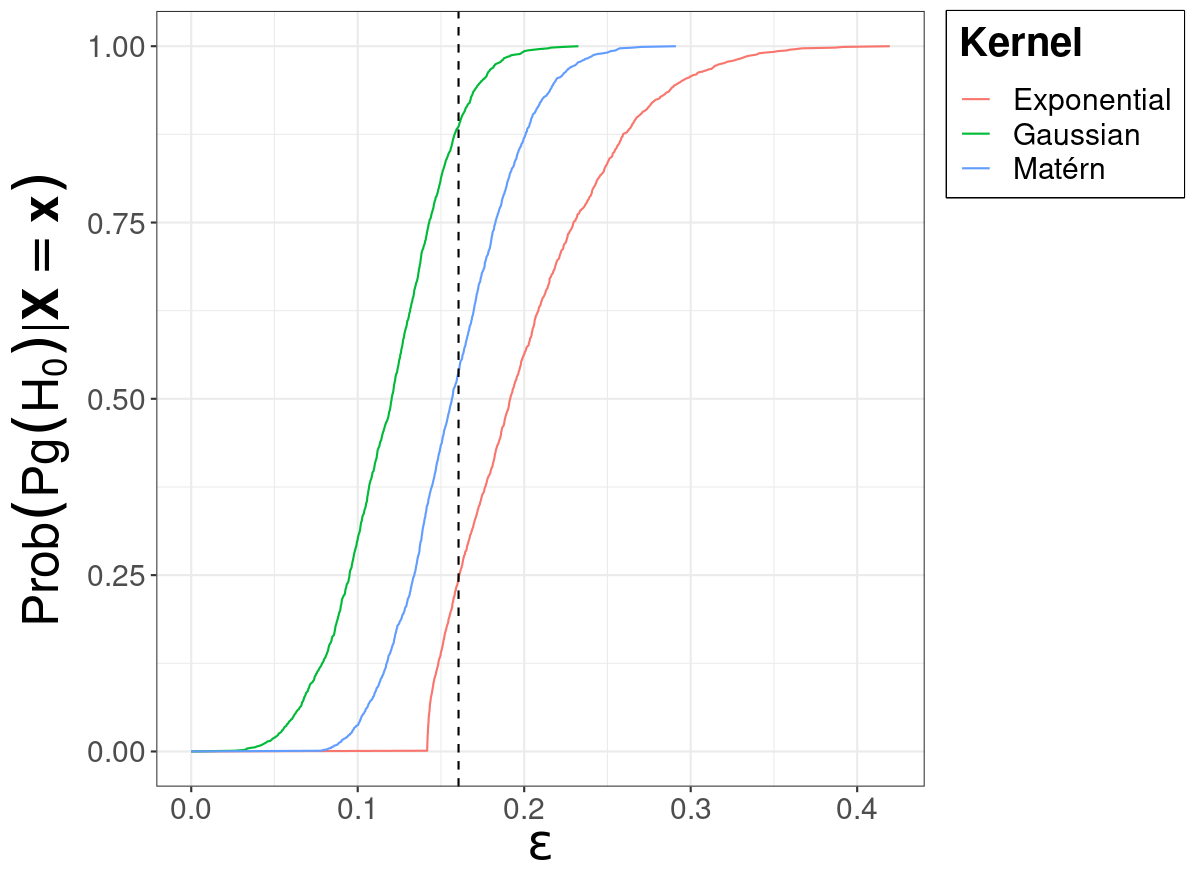}
             \caption{$H_0^1: a(t) = \beta_0 + \beta_1 t$}
             \label{fig:drop1}
         \end{subfigure}
         \hfill
         \begin{subfigure}[b]{0.49\textwidth}
             \centering
             \includegraphics[width=\textwidth]{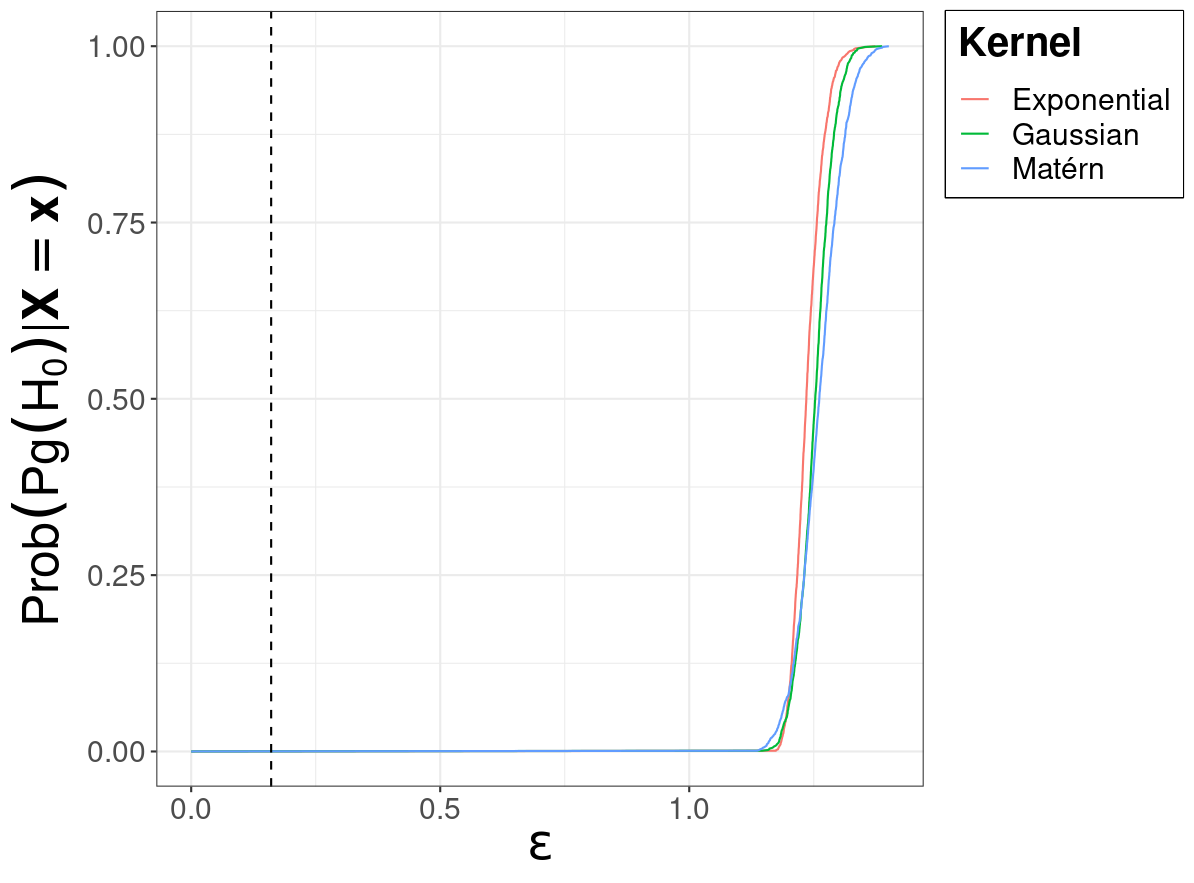}
             \caption{$H_0^2: a(t) = \beta_0$}
             \label{fig:drop2}
         \end{subfigure}
            \caption{Largest $\varepsilon$ that entails rejection and the posterior probability of $Pg(H_0)$ for each kernel of the Gaussian process in \autoref{ex:drop12}. The dashed line marks the threshold value ($\varepsilon \approx 0.1606$).}
            \label{fig:drop-test}
    \end{figure}
\end{customex}

Going beyond linear regression, \autoref{thm:lm_inf} can also be used for testing models whose regression function depends on a linear combination of $\boldsymbol{b}(\boldsymbol{x})$, such as GLMs. For a known function $h(\cdot)$, the same test can be performed by switching the hypothesis in \autoref{eq:lin_hyp} for
\begin{equation}
    \label{eq:glm_hyp}
    H_0: R(\boldsymbol{x}) = h(\boldsymbol{b}(\boldsymbol{x})\boldsymbol{\beta}) \Longleftrightarrow H_0: h^{-1}(R(\boldsymbol{x})) = \boldsymbol{b}(\boldsymbol{x})\boldsymbol{\beta}, \quad \forall \boldsymbol{x} \in \mathcal{X}, \quad \boldsymbol{\beta} \in \mathbb{R}^k,
\end{equation}
as long as $h^{-1}(\cdot)$ can be obtained. In \autoref{sct:binary}, we use this strategy in a simulated setting to check for adherence when the response variable is binary.

\subsection{Goodness-of-fit test}
\label{sct:distribution}

Let $H_0: X \sim F$, where $F$ is a fixed distribution function. Then, for a threshold $\varepsilon$ and a dissimilarity function $d(\cdot,\cdot)$, the pragmatic hypothesis is given by
\begin{equation}
    Pg(H_0) = \left\{P \in \mathbb{F}: d(F,P) <\varepsilon \right\},
    \label{eq:dist1}
\end{equation}
since $F$ is the only distribution function that belongs to $H_0$. Thus, a goodness-of-fit test can be executed through \mytest{}, with the problem of the dissimilarity being reduced to that of obtaining $d(F, \cdot)$.

Going beyond a single distribution function, we can also determine the pragmatic hypothesis for a parametric family. If $\theta \in \Theta$ is the parameter vector of such family, the null hypothesis is
\begin{equation*}
    H_0: X \sim F_\theta, \quad \theta \in \Theta_0 \subseteq \Theta.
\end{equation*}
Hence, the pragmatic hypothesis is
\begin{equation}
    Pg(H_0) = \left\{P \in \mathbb{F}: \inf_{\theta \in \Theta_0} d(F_{\theta},P) <\varepsilon \right\}.
\label{eq:dist2}
\end{equation}
This means that the process of identifying if a candidate $P \in \mathbb{F}$ belongs to $Pg(H_0)$ can be translated into an optimization procedure. For every given $P$, the objective is to find $\hat{\theta} \in \Theta_0$ such that $d(F_{\hat{\theta}},P) \le d(F_{\theta},P), \forall \theta \in \Theta_0$. Then, if $\hat{\theta}$ provides a dissimilarity smaller than $\varepsilon$, we conclude that $P \in Pg(H_0)$.

\begin{example}[$H_0: N(t), t \in \mathbb{R}^+,$ is a Poisson process]
\label{ex:pois_process}

The Poisson process \citep{ross2009} is a counting process that assumes that $N(t) \sim Poisson(\lambda t), \forall t \in \mathbb{R}_{\ge 0}$. Let $(X_1, \cdots, X_n)$ be a sample of the moment in time each observation has occurred and $T_i := X_i - X_{i-1}$, $i \in \{1, \cdots, n\}$. Then,
$$
    H_0: N(t) \text{ is a Poisson process} \Longleftrightarrow H_0: T_i|\lambda \stackrel{ind.}{\sim} Exp(1/\lambda), \lambda \in \mathbb{R}_{\ge 0}, i \in \{1, \cdots, n\},
$$
and hence the pragmatic hypothesis is
\begin{equation}
    \label{eq:prag_pois}
    Pg(H_0) = \left\{P \in \mathbb{F}_T: \inf_{\lambda \in \mathbb{R}_{\ge 0}} d(F_{\lambda},P) <\varepsilon \right\},
\end{equation}
where $F_{\lambda} \equiv Exp(1/\lambda)$.

Choosing the $L^{\infty}$ distance \textendash{} the same used in the Kolmogorov-Smirnov test \citep{shiryayev1992} \textendash{} for \eqref{eq:prag_pois}, it would be represented as
\begin{equation}
\label{eq:linf_pois}
    d_\infty(F_{\lambda},P) = \sup_{t \in \Omega}|F_{\lambda}(t) - P(t)| = \sup_{t \in \Omega}|1 - \exp(-\lambda t) - P(t)|.
\end{equation}
Therefore, $P \in Pg(H_0) \Longleftrightarrow \exists \lambda \in \mathbb{R}_{\ge 0}: d_\infty(F_{\lambda}, P) < \varepsilon$.
Since $P$ is fixed, such condition can be verified through an optimization procedure by finding the value for $\lambda$ that minimizes $d_\infty(F_{\lambda}, P)$, which is achievable through general optimization routines such as the \texttt{optim} function in \texttt{R} \citep{rcore2022}. This exact test is carried out in \autoref{sct:neuron_test1}.
\end{example}

\subsection{Quantile test}
\label{sct:quantile}

In this section, we propose a quantile test that does not require any distributional assumption on the data. Let $x_0$ and $p_0$ be such that $\mathbb{P}(X \le x_0) = p_0$, i.e., $x_0$ is the $p_0$-quantile of $X$ if $\mathbb{P}$ is its true probability measure. Then, the hypothesis of interest for this case would be
$$
    H_0: F(x_0) = p_0, \quad F \in \mathbb{F}_X.
$$
Closed-form solutions for this hypothesis depend on the dissimilarity function of choice. Let
\begin{equation}
    \label{eq:l1_dist}
    d_1(F, G) := \|F-G\|_1 = \int_{\mathbb{R}} |F(x)-G(x)|dx, \quad F, G \in \mathbb{F}.
\end{equation}
The following theorem provides a straightforward procedure for obtaining the pragmatic hypothesis when \eqref{eq:l1_dist} is the dissimilarity function.
\begin{theorem}[Quantile test]
\label{thm:quant_inf}
Let $H_0: F(x_0) = p_0$, $F \in \mathbb{F}_X,$ be the null hypothesis and take \eqref{eq:l1_dist} as the dissimilarity function. Then, $\forall P \in \mathbb{F}_X$, if $a := \min(P^{-1}(p_0), x_0)$ and $b := \max(P^{-1}(p_0), x_0)$,
\begin{equation}
    \label{eq:quant_inf}
    \inf_{P_0 \in H_0} d(P_0, P) = \inf_{P_0 \in H_0} \int_{-\infty}^\infty |P_0(x) - P(x)|dx = \int_{a}^b |p_0 - P(x)|dx.
\end{equation}
\end{theorem}

If for some reason \eqref{eq:quant_inf} cannot be analytically obtained, a Monte Carlo integration procedure \citep{robert2005} could be used. This result is applied in \autoref{sct:neuron_test2}.

\subsection{Two-sample test}
\label{sct:two_sample}

In this section, we provide a pragmatic version of the nonparametric two-sample test, a test whose hypothesis originally states that the true distribution functions of two different datasets are the same. In other words, if $X$ and $Y$ are the random variables of interest and $F_X$ and $F_Y$ are their respective distribution functions, then
$$
    H_0: F_X = F_Y = F, \quad F\in \mathbb{F},
$$
is the hypothesis which we seek to expand.

We highlight that $\mathbb{H} = \mathbb{F} \times \mathbb{F}$, i.e., the hypothesis space is the Cartesian product of the space of distribution functions. In \autoref{fig:2sample_hyp}, a visualization is provided to give an idea of the peculiarities of such space. Each axis of the figure represents the distribution function of a specific population. Then, the green line represents the null hypothesis that both distributions are equal. Thus, while the red dot is an element of $\mathbb{H}$ (i.e., a given pair of distribution functions), the red arrow represents the smallest distance between such element and $H_0$.

\begin{figure}[ht!]
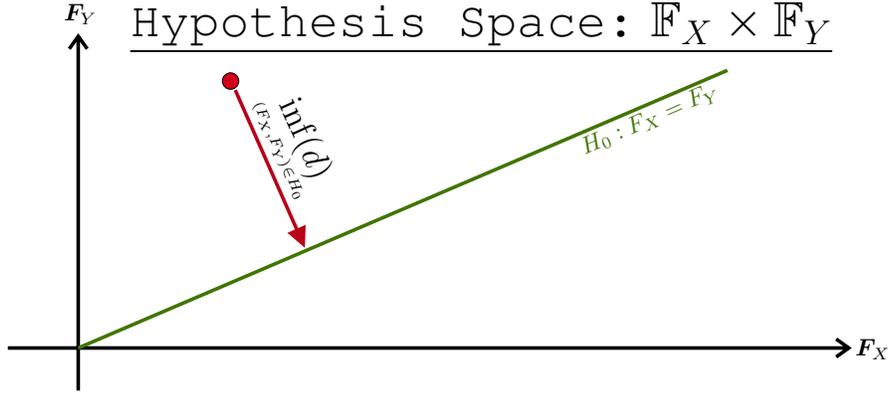

    \centering
        \resizebox{.9\linewidth}{!}{\TwoSample}
    \caption{Representation of the nonparametric two-sample hypothesis. The green line is the original null hypothesis, while the red dot is a pair of distribution functions.}
    \label{fig:2sample_hyp}
\end{figure}

The following result provides an analytical solution for the infimum that is solely based on the distance between the functions obtained from the data:
\begin{theorem}[Two-sample test]
\label{thm:2sample_inf}
Let $H_0: F_X = F_Y = F$, $F \in \mathbb{F}$, be the null hypothesis, and $(P_X, P_Y)$ be a pair of distribution functions. If $d(\cdot,\cdot)$ is such that
\begin{equation}
    \label{eq:2sample_cond}
    d[(F_X, F_Y), (P_X, P_Y)] = d^*(F_X, P_X) + d^*(F_Y, P_Y),
\end{equation}
where $d^*(\cdot, \cdot)$ is a distance function, then
$$
    \inf_{(F_X, F_Y) \in H_0} d[(F_X, F_Y), (P_X, P_Y)] = d^*(P_X, P_Y).
$$
\end{theorem}

More than simply identifying the infimum for a given dissimilarity, \autoref{thm:2sample_inf} provides a solution that works for any distance function while keeping the intuitive appeal of reaching a decision solely based on the discrepancy between the distribution functions of $X$ and $Y$. Such appeal can be observed in both classical statistical tests -- such as the Kolmogorov-Smirnov test \citep{shiryayev1992} -- and more recent iterations \citep{inacio2020, ceregatti2021}. Moreover, our version can be seen as an enhancement of the Kolmogorov-Smirnov test, since it allows for the choice of other distance functions and takes model uncertainty into account.

Since the theorem makes no restriction on the choice of the distance function, the classification dissimilarity (\autoref{def:class_dissim}) could be used in this case if we subtract it by 0.5, i.e.,
\begin{equation}
    d_C^*(F_X, F_Y) = 0.5 \left[\mathbb{P}\left(\frac{f_X(Z)}{f_Y(Z)} > 1 \Big| Z \sim F_X \right) + \mathbb{P}\left(\frac{f_Y(Z)}{f_X(Z)} > 1 \Big| Z \sim F_Y \right) \right] - 0.5,
    \label{eq:np_dist}
\end{equation}
where $f_X$ and $f_Y$ are the respective density functions of $F_X$ and $F_Y$. \eqref{eq:np_dist} is the distance function used in the simulated study (\autoref{sct:NvsT}).

\section{On choosing the threshold \texorpdfstring{$\boldsymbol{\varepsilon}$}{TEXT}}
\label{sct:epsilon}

The current lack of standards and guidelines for establishing the threshold $\varepsilon$ is the main drawback for researchers that seek to enlarge their hypotheses, so it is imperative to derive suggestions for $\varepsilon$ that can be more generally applied. Although some solutions have been proposed to specific problems \citep{hodges1954, hobbs2007, gross2014, kruschke2018, lakens2018}, none of them offer strategies for determining the threshold in more general settings, such as when dealing with nonparametric hypotheses.

Although we provide general suggestions on how to choose $\varepsilon$ based on the type of intuition a researcher has, these suggestions serve more as a starting point for discussions. Ideally, the value of $\varepsilon$ should reflect a utility judgement of the researcher, their notion of what results should be indistinguishable from the null hypothesis in practice.

\subsection{Intuitions that lead to \texorpdfstring{$\boldsymbol{\varepsilon}$}{TEXT}}
\label{sct:one_eps}

We begin by presenting suggestions that, if followed, are assertive enough to establish a unique value for $\varepsilon$. They consist of:

\paragraph{Using theory or measurement errors.} In this case, there is external information available to determine $\varepsilon$, coming either through theoretical assumptions, knowledge of measurement errors or both. The scope of possible dissimilarity functions for this case would then be limited to those that can use the information on $\varepsilon$ to their advantage.

\begin{customex}{1.3}[Water droplet experiment, continued]\label{ex:drop13}
    This last part of the example uses known results of Physics and more details from the original experiment \citep{duguid1969} to determine a value for $\varepsilon$.
    
    While the objective of the study is to evaluate the radius of droplets through time, the radius itself was not measured directly. Instead, Stoke's law was used to estimate it based on the velocity. It states that
    \begin{equation}
        V_T(t) = \frac{a(t)^2}{K_s} \Longrightarrow a(t) = \sqrt{V_T(t) \times K_s},
        \label{eq:stokes}
    \end{equation}
    where $V_T$ is the terminal velocity and $K_s$ is a known constant that depends on factors such as temperature and humidity (in this case, $K_s = 8.446$). However, $V_T$ was not registered in the experiment, with the mean velocity ($V_M$) being used instead since it can be inferred from the pictures of the camera. Still, since $V_T$ is the derivative of the droplet's position through time, it can be estimated through symmetric differences of the droplet's position in a weighted least squares regression \citep{wang2015}.

    The value of $\varepsilon$ must take into account all sources of error in the experiment. Switching $V_T$ for $V_M$ in \eqref{eq:stokes} is the first source of measurement error, while the second comes from the measurement of $V_M$ itself. A square grid was positioned in front of the device to aid in registering the droplet's position in the tube at a give time, leading to a measurement error of at most $\delta = 0.14$. Therefore, if $\eta := \max_{t \in T} |V_T(t) - V_M(t)|$ is the first measurement error, then $\exists h \in [-(\delta + \eta), (\delta + \eta)]$ such that
    \begin{equation}
      a(t) = \sqrt{K_sV_T(t)} = \sqrt{K_s(V_M(t) + h)}, \quad t \in T.
    \end{equation}
    If $\hat{a}(t) = \sqrt{K_sV_M(t)}$ is the estimate of the radius at time $t \in T$, the margin of error of the radius is
    \begin{align*}
        \epsilon :&= \max_{t \in T, h \in [-(\delta + \eta), (\delta + \eta)]}|a(t) - \hat{a}(t)|\\
        &= \max_{t \in T}\left\{\left|\sqrt{K_s(V_M(t) - \delta - \eta)} - \hat{a}(t)\right|, \left|\sqrt{K_s(V_M(t) + \delta + \eta)} - \hat{a}(t)\right|\right\}.
    \end{align*}

    \autoref{fig:diff_data} shows the respective errors for both $|V_T(t) - V_M(t)|$ and $|a(t) - \hat{a}(t)|$. Following the suggestions of \cite{wang2015}, we use $k = 4$ symmetric differences and disregard the first and last $k$ estimates of $V_T$, leading to $\eta = \max_{t \in T}|V_T(t) - V_M(t)| \approx 0.3555$. By plugging $\eta$ in $|a(t) - \hat{a}(t)|$, we conclude that $\epsilon = \max_{t\in T} |a(t) - \hat{a}(t)| \approx 0.6218$.
    
    \begin{figure}[ht]
        \centering
        \includegraphics[width = \linewidth]{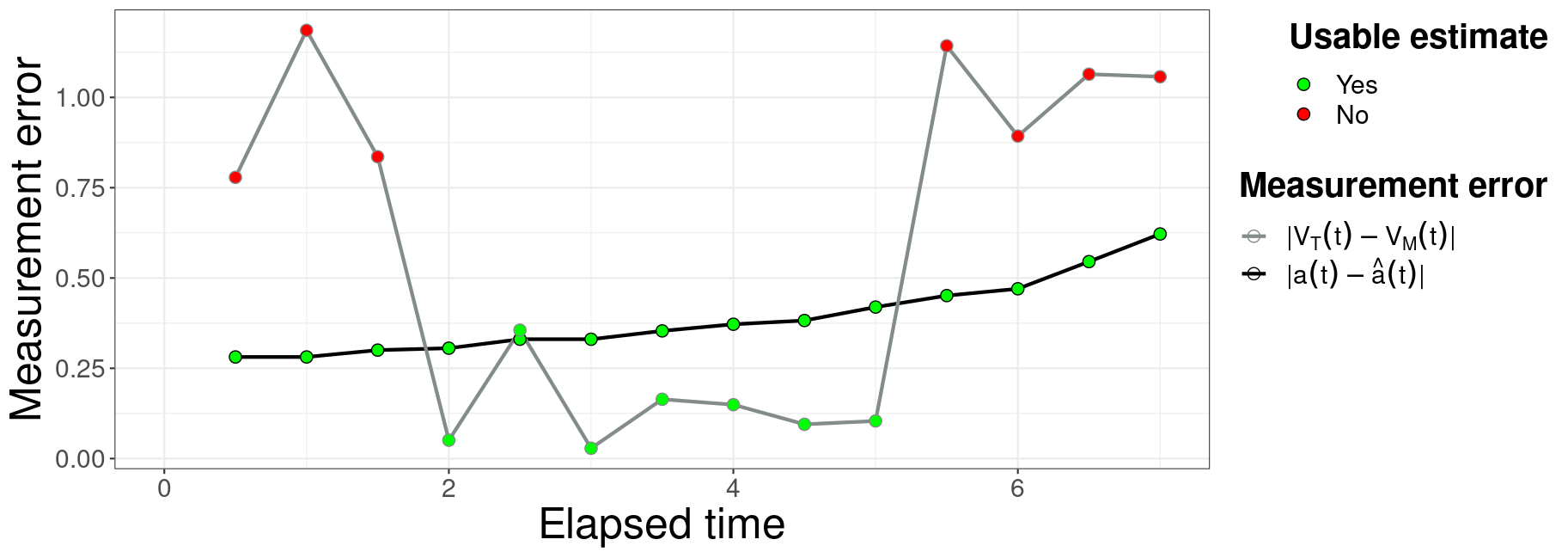}
        \caption{Measurement errors of velocity and radius. Red dots represent poor estimates of $|V_T(t) - V_M(t)|$, green dots represent better ones. The largest green estimate of $|V_T(t) - V_M(t)|$ was the one plugged into $|a(t) - \hat{a}(t)|$.}
        \label{fig:diff_data}
    \end{figure}

    The last step required for reaching the threshold $\varepsilon$ is to adapt $\epsilon$ -- which is related to the $l^\infty$ distance -- to the dissimilarity function of interest, a weighted version of the $l^2$ distance. Proposition 6.11 of \cite{folland2013} ensures that $l^2 \subset l^\infty$, therefore
    $$
        \inf_{\beta \in \mathbb{R}^p}\sqrt{\frac{1}{n}\sum_{i = 1}^n(\boldsymbol{x}_i'\boldsymbol{\beta} - g(\boldsymbol{x}_i))^2} \le \sqrt{\frac{1}{n}}\epsilon \Longrightarrow \inf_{\beta \in \mathbb{R}^p} \max_{i \in \{1, 2, \cdots, n\}}|\boldsymbol{x}_i'\boldsymbol{\beta} - g(\boldsymbol{x}_i)| \le \epsilon.
    $$
    Thus, using $\varepsilon = \sqrt{\frac{1}{n}}\epsilon \approx 0.1606$ as the threshold for the $l^2$ distance leads to the same conclusion as using $\epsilon$ for the $l^\infty$ distance when not rejecting the hypothesis.
\end{customex}

\paragraph{Setting the threshold through the prior.} While specifying a probability mass to a point null hypothesis might be ill-advised if that does not represent the researcher's belief \citep[page 21]{berger1985}, attributing a prior probability to the pragmatic hypothesis itself is a valid possibility. If that is the case, we can then use this to obtain the threshold by checking which value $\varepsilon$ should assume to match the prior. Formally, let $\mathbb{P}(\mathcal{P} \in Pg(H_0)) = \delta$ be the prior probability of $Pg(H_0)$ being true. Since
$$
    \mathbb{P}(\mathcal{P} \in Pg(H_0)) = \mathbb{P}\left(\inf_{P_0 \in H_0} d(P_0, \mathcal{P}) < \varepsilon \right) = \delta \Longleftrightarrow Q_\delta \left(\inf_{P_0 \in H_0} d(P_0, \mathcal{P})\right) = \varepsilon,
$$
where $Q_\delta(\cdot)$ is the $\delta$-quantile function, then $\varepsilon$ is uniquely determined through the choice of $\delta$ and the prior over $\mathcal{P}$.

When the prior over $\mathcal{P}$ is informative but a value for $\delta$ is not clear, we can use the fact that the prior uncertainty is greater than the posterior uncertainty to our advantage. By taking $\delta = \alpha$, it is expected that $Q_\delta\left(\inf_{P_0 \in H_0} d(P_0, \mathcal{P} | \boldsymbol{X} = \boldsymbol{x})\right)$ should be smaller than $\varepsilon$ when $Pg(H_0)$ is true and greater when it is false. This suggestion is further explored in \autoref{sct:NvsT}.

\paragraph{Building from related studies.} Say that there is at least one study in the literature with positive results which can be used as reference for your own study. Since the interest here is to provide a direct comparison between their findings and yours, apply the same model and the same significance level $\alpha$ of your study to their data, choosing the smallest $\varepsilon$ that leads to non-rejection. If there are multiple studies, take the largest $\varepsilon$ between them so that none of the studies is rejected.

This approach is particularly useful for reproducibility research, since newer studies tend to have a larger sample and data with higher quality than the old one, so the same conclusion should be reached if the hypothesis is true. Other cases where this approach might be reasonable are when there has been observed an effect for a given group (geographical region, social class, species, etc.) and we wish to check if the same effect exists for a different group. A similar idea is found in \citet[Section 9.12]{lakens2022}

\begin{example}[Worldwide gender wage gap]
\label{ex:wage_gap}
The gender wage gap is a multifaceted issue that remains harming women in the workforce for the last 200 years \citep{goldin1990}, even though some advances have been made to reduce it \citep{blau2017}. Let $X$ represent the difference between the wage gap of two consecutive years and let the null hypothesis be
$$
    H_0: F_X(0) = 0.25, \quad F_X \in \mathbb{F}_X,
$$
i.e, that only 25\% of the countries have managed to reduce the wage gap between years. Using data from the ``pay gap as difference in hourly wage rates'' in different countries between the years of 2021 and 2020 \citep{unece}, our objective is to deliberate what $\varepsilon$ should be used in a follow-up study on the same matter.

We remove the countries with one or both entries missing, resulting in a sample of $n = 28$. Then, we use a Dirichlet process \citep{ferguson1973} with scaling parameter equal to 1 and centered on $N(0, 10)$ as the prior. Lastly, we apply \autoref{thm:quant_inf} and choose $\varepsilon$ as the largest value that would lead to rejecting $H_0$ when $\alpha = 0.05$ on \mytest{}, leading to $\varepsilon \approx 0.0312$. Therefore, in a follow-up study, if such value of $\varepsilon$ leads to rejection, we can safely conclude that $H_0$ has failed to reproduce.
\end{example}

\subsection{Intuitions that delimit \texorpdfstring{$\boldsymbol{\varepsilon}$}{TEXT}}
\label{sct:many_eps}

When the intuitions provided by the researcher are not sufficient to provide a definitive value for $\varepsilon$, but can nevertheless be of use, some suggestions are:

\paragraph{Setting an upper bound through examples.} This case consists of listing the pairs of elements in the hypothesis space that the researcher assumes to be negligible from each other. Then, by obtaining the dissimilarities of those combinations and taking the largest of them, the result can be assigned as the value of $\varepsilon$. This represents a lower bound for the real $\varepsilon$ of interest and, in case the test does not reject the hypothesis, provides the exact same conclusion as the ``true'' $\varepsilon$. This strategy is employed in \autoref{sct:application}.

\paragraph{Using multiple candidates for \texorpdfstring{$\boldsymbol{\varepsilon}$}{TEXT}.} We assume that, instead of dealing with a unique $\varepsilon$, there is a list or a range of values for $\varepsilon$ which one must consider. This might happen when there are multiple professionals and each of them provides their own suggested $\varepsilon$, such as when there are more ``liberal'' or ``conservative'' choices available for it \citep{gross2014}. The idea is to simply perform \mytest{} for each $\varepsilon$ on the list (or to a grid based on the range of reasonable candidates) and take as final the decision that came out the most. Further still, we could weight each candidate based on some criteria (such as the importance of the professional or how much smaller a specific $\varepsilon$ is when compared to the others) and apply the same idea. For example, since the classification dissimilarity (\autoref{def:class_dissim}) only takes values in $[0.5, 1]$, one could build a grid and use weights that decrease linearly to reach a decision, such as giving weight 1 to $\varepsilon = 0.5$ and 0 to $\varepsilon = 1$.

\paragraph{Direct graphical evaluation.} Lastly, we suggest the user to simply plot the conclusion as a function of $\varepsilon$ and $\alpha \in [0,1]$, and then use this graphical evaluation to decide if rejecting the hypothesis makes sense. This is the only suggestion that does not require setting neither $\varepsilon$ nor $\alpha$ beforehand and should thus be used with caution. After all, this liberty could influence the analyst of the test towards making the conclusion they already agree with, biasing the results.

More than an actual suggestion for reaching conclusions, this plot acts as a tool for transparency and plurality. Since disagreements on the choice of $\alpha$ and $\varepsilon$ are sure to be common, it neatly provides an indication of the decision one should take for their particular choice without the requirement of doing the whole analysis once again.

\begin{customex}{3.1}[Worldwide gender wage gap, continued]
    \autoref{fig:alpha-eps} presents, for each combination of $(\varepsilon, \alpha)$, the decision suggested by \mytest{}, with the red area implying rejection and the green area implying non-rejection. Based on the figure alone, we can safely conclude that researchers advocating for $\varepsilon \ge 0.1$ should not reject the hypothesis, since $\alpha$ would need to be set around 0.5 to lead to rejection.

    \begin{figure}[ht]
        \centering
        \includegraphics[width = \linewidth]{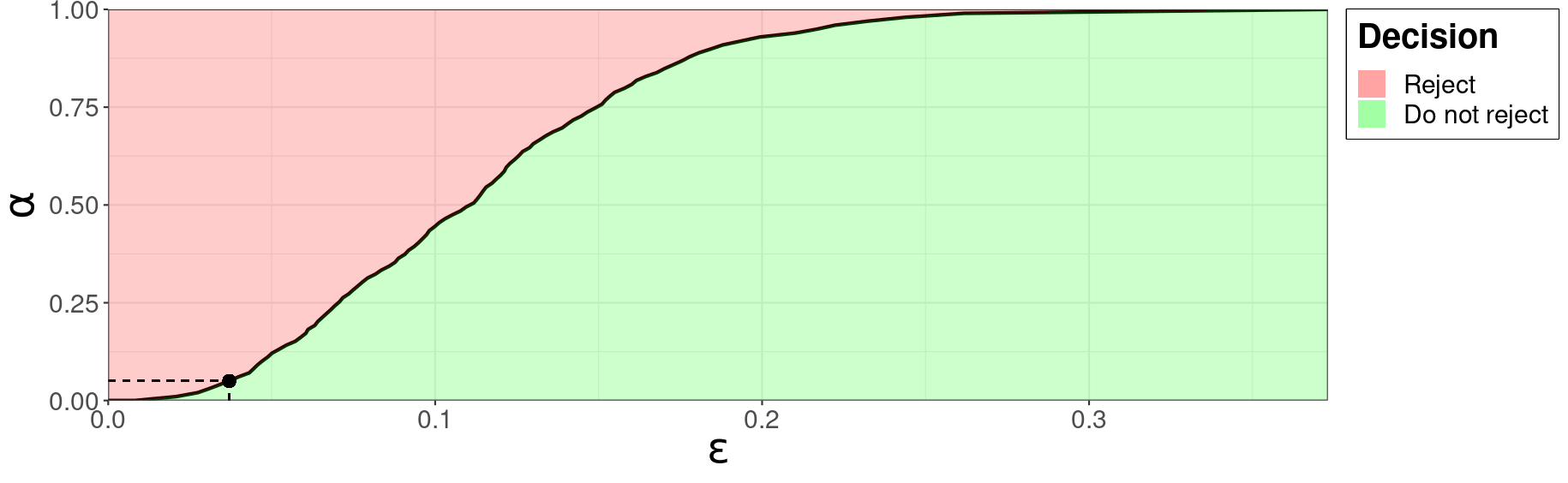}
        \caption{Decision regions as a function of $(\varepsilon, \alpha)$ for the gender wage gap data (red for rejection, green for non-rejection). The black dot in the curve indicates the initial choice for $\varepsilon$ based on $\alpha = 0.05$.}
        \label{fig:alpha-eps}
    \end{figure}
\end{customex}

\section{Simulated studies}
\label{sct:simul}

\subsection{Regression on a binary response variable}
\label{sct:binary}
This next setting uses data from a logistic regression to evaluate if the test can discriminate between link functions as the sample size grows. Let $\boldsymbol{X}$ be a 3 column matrix, with all values sampled from a $U(-3,3)$, and $Y$ be a binary variable such that
\begin{equation*}
    \mathbb{P}(Y_i = 1| \boldsymbol{X}) = \frac{1}{1 + \exp(-0.5 +1.5\boldsymbol{X}_{i,1}  -2\boldsymbol{X}_{i,2} +0\boldsymbol{X}_{i,3})}, \quad i \in \{1, 2, \cdots, n\},
\end{equation*}
where $n$ represents the sample size.

We use the nonparametric model proposed by \cite{deyoreo2015} to draw estimates of $\mathbb{P}(Y = 1 | \boldsymbol{X})$ and then apply the hypothesis from \autoref{eq:glm_hyp} to check which link function (logit or cloglog) seems better suited for the data. The prior specification follows the second approach suggested by that paper and we set $Gamma(2,2)$ as the prior distribution of the Dirichlet process' scaling parameter. We truncate the Dirichlet process so that it provides 30 mixture components.

\autoref{fig:link-test} provides the test results for different sample sizes. We observe that, when increasing the sample size, the value of $\varepsilon$ that would lead to rejection becomes consistently smaller for both link functions and the decision becomes less dependent on the choice of $\alpha$. Still, the logit link presents a superior performance for all sample sizes, and the difference between curves becomes more apparent as well.

\begin{figure}[ht!]
     \centering
     \begin{subfigure}[b]{0.49\textwidth}
         \centering
         \includegraphics[width=\textwidth]{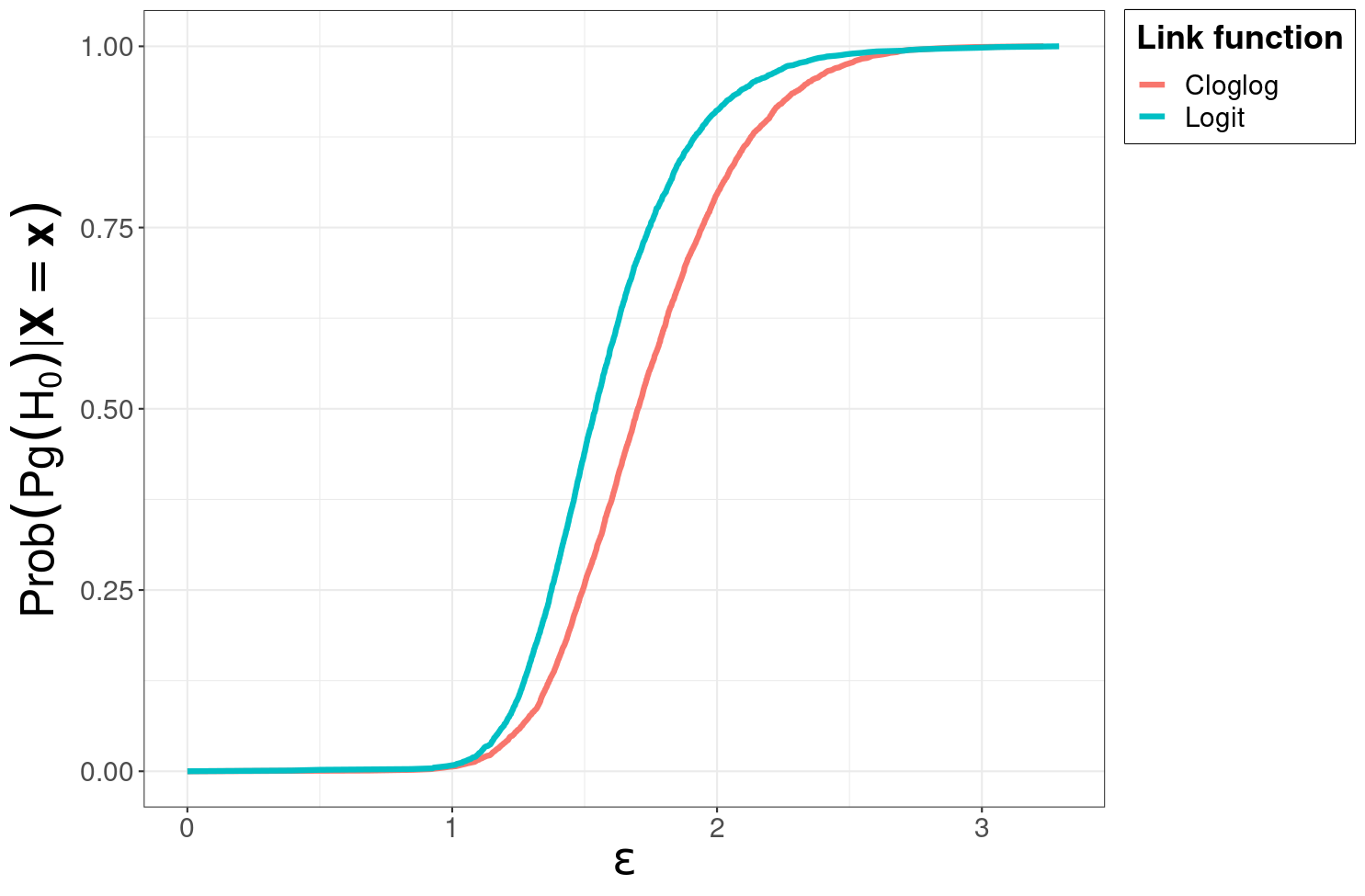}
         \caption{$n = 2500$}
         \label{fig:logit25}
     \end{subfigure}%
     \hfill
     \begin{subfigure}[b]{0.49\textwidth}
         \centering
         \includegraphics[width=\textwidth]{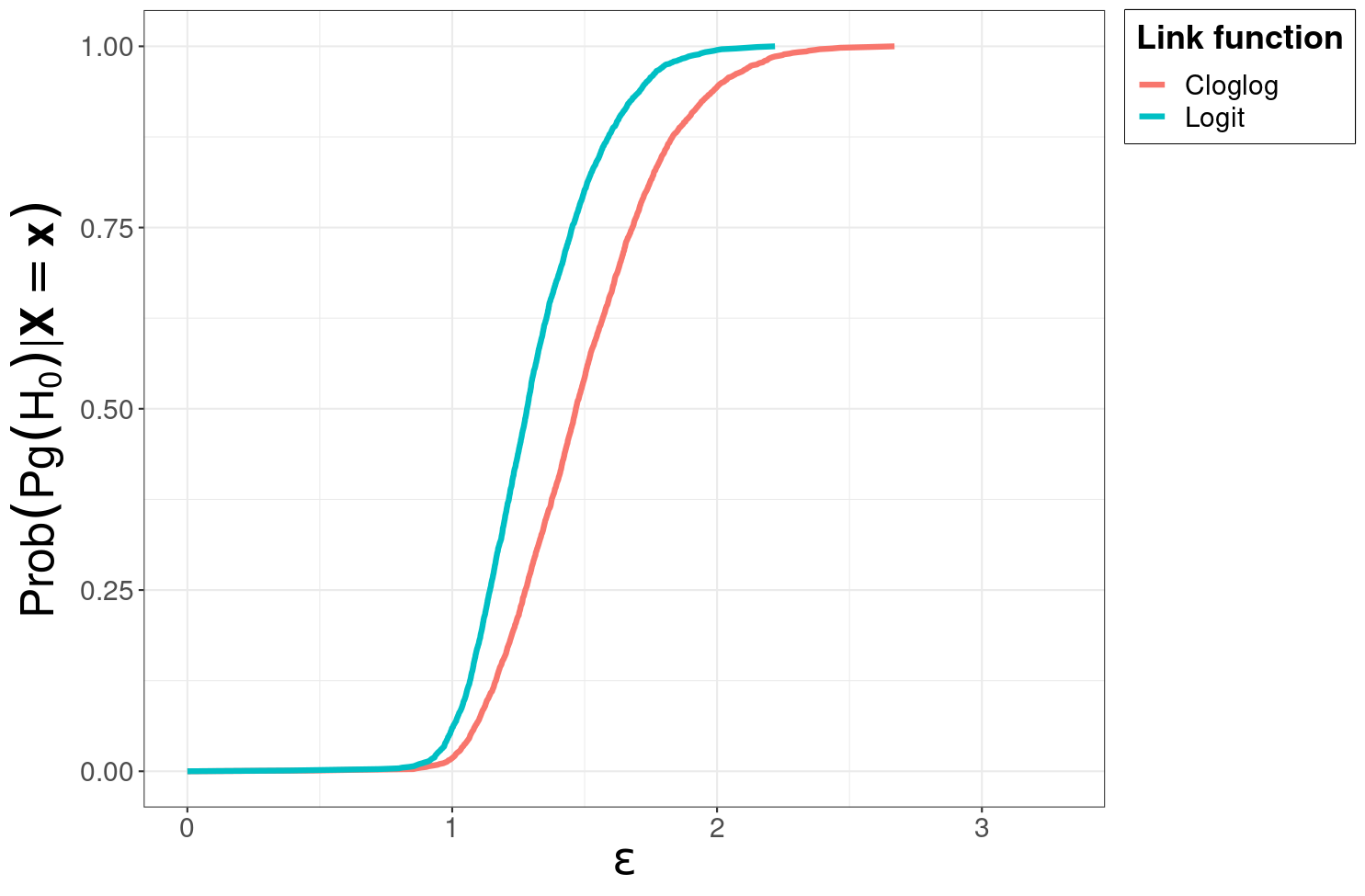}
         \caption{$n = 5000$}
         \label{fig:logit50}
     \end{subfigure}
     \begin{subfigure}[b]{0.49\textwidth}
         \centering
         \includegraphics[width=\textwidth]{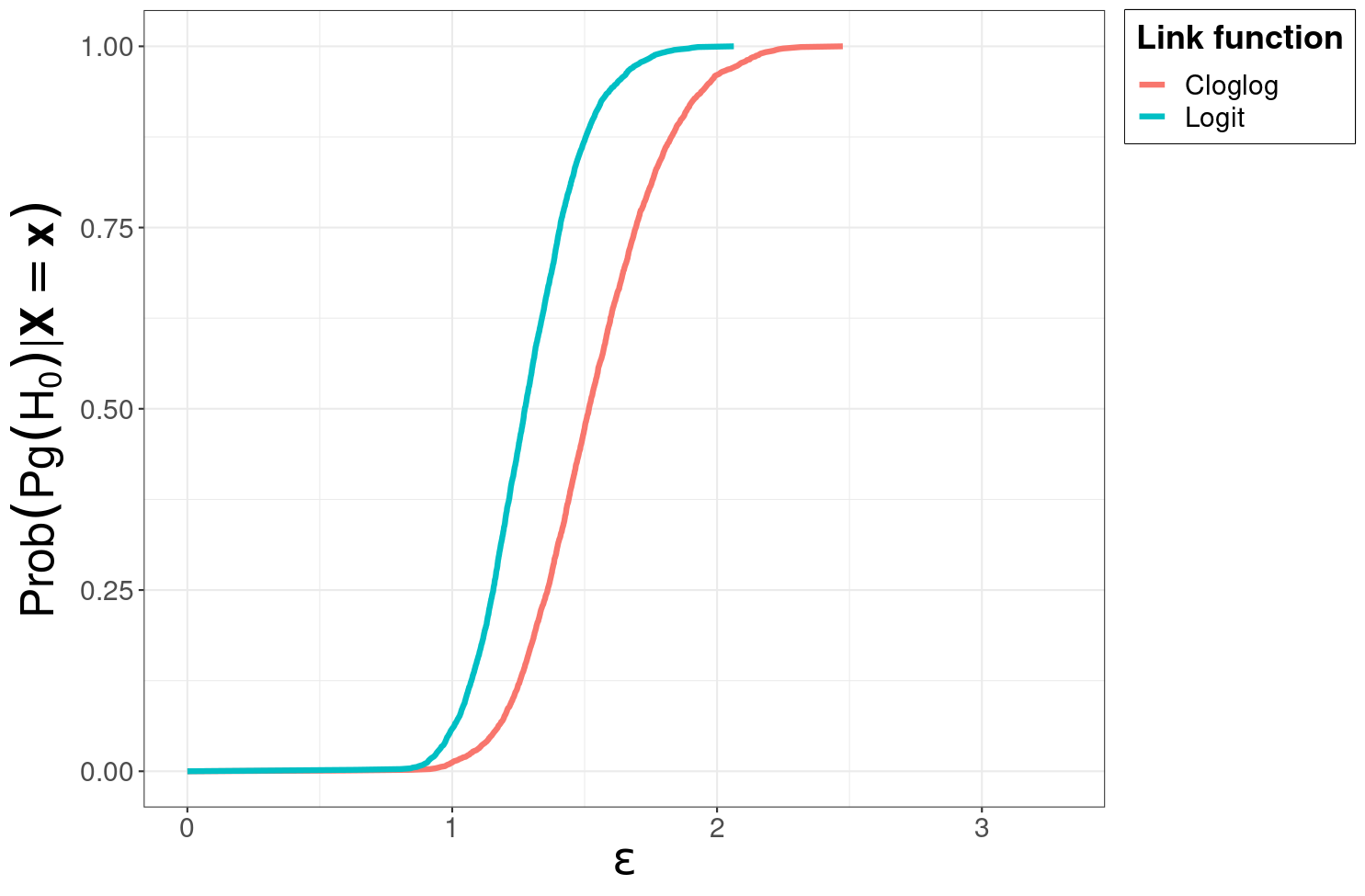}
         \caption{$n = 7500$}
         \label{fig:logit75}
     \end{subfigure}%
     \hfill
     \begin{subfigure}[b]{0.49\textwidth}
         \centering
         \includegraphics[width=\textwidth]{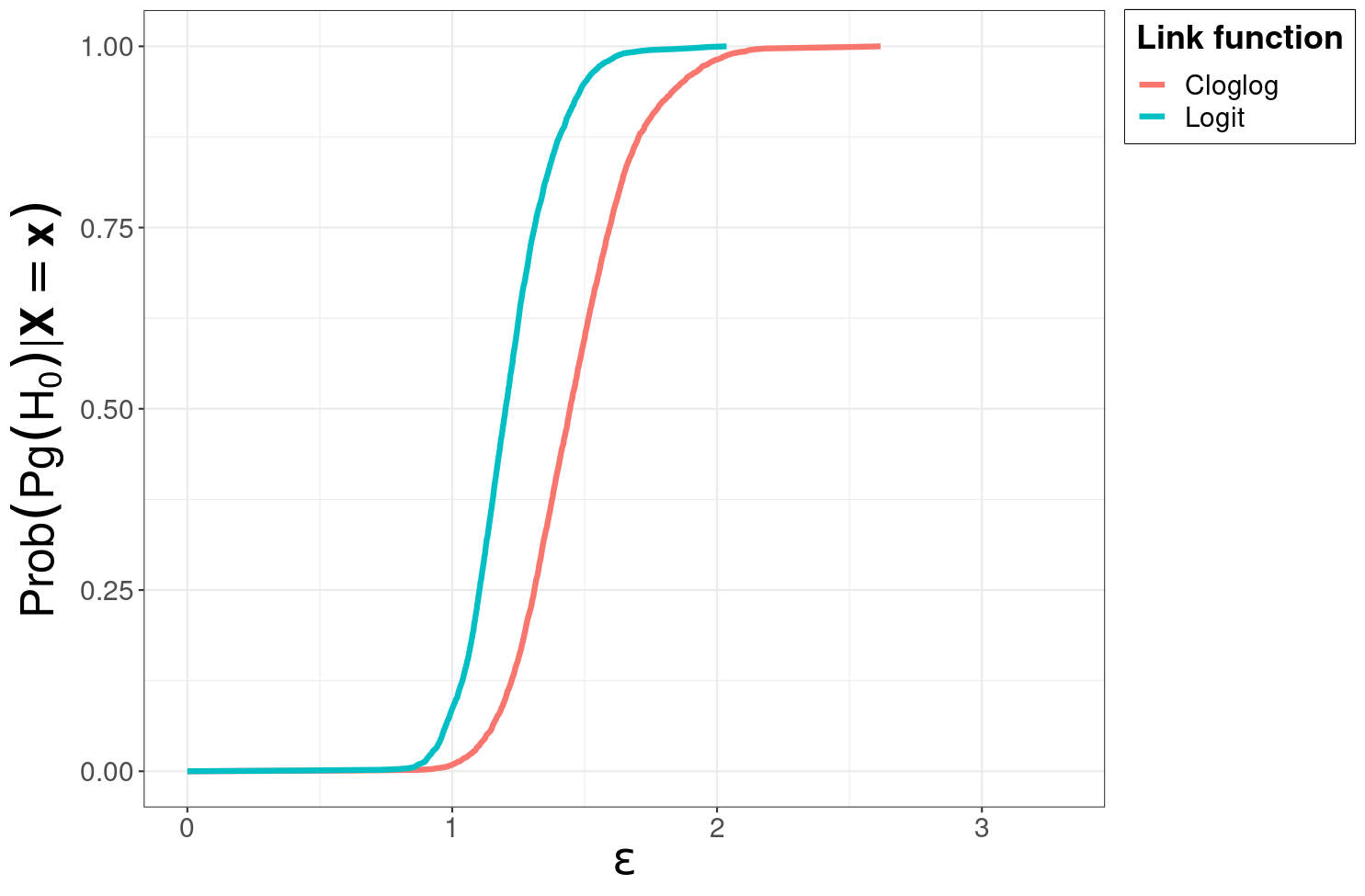}
         \caption{$n = 10000$}
         \label{fig:logit100}
     \end{subfigure}
        \caption{Largest $\varepsilon$ that entails rejection and the posterior probability of $Pg(H_0)$ for each link function and multiple sample sizes.}
        \label{fig:link-test}
\end{figure}

\subsection{Similarity between Normal and Student's t distributions}
\label{sct:NvsT}

Distribution tables have been widely present in statistical textbooks through time \citep{fisher1963, casella2001} and are still used nowadays for pedagogical purposes \citep{mitchell2018}. Particularly for the Student's t distribution table, a common feature is that the table becomes sparser after 30 degrees of freedom, implying that after 30 the deviations between the quantiles are deemed as negligible. Moreover, since the Student's t distribution converges to a standard Normal as the degrees of freedom tend towards infinity, some claim that using the Normal distribution as an approximation when the degrees of freedom are over 30 is good enough for most practical purposes \citep{pett2016}. We use this ``consensus'' as the basis for our simulation study, verifying how sensitive \mytest{} can be to it.

\begin{table}[ht]
    \centering
    \begin{tabular}{r|llll|llll}
        \hline
        \multirow{ 2}{*}{Sample size}& \multicolumn{4}{c}{\mytest{}} & \multicolumn{4}{c}{PTtest}\\
        & $c = 1$ & $c = 4$ & $c = 7$ & $c = 10$ & $c = 1$ & $c = 4$ & $c = 7$ & $c = 10$ \\ 
        \hline
        $10^2$ & 0 & 0.4680 & 0.8421 & 0.8856 & 0.9509 & 0.8843 & 0.8478 & 0.8174 \\ 
        $10^3$ & 0.0057 & 0.9998 & 1 & 1 & 1 & 1 & 1 & 1 \\ 
        $10^4$ & 1 & 1 & 1 & 1 & 1 & 1 & 1 & 1 \\ 
        $3 \times 10^5$ & 1 & 1 & 1 & 1 & 1 & 0.9999 & 1 & 1 \\
        $4 \times 10^5$ & 1 & 1 & 1 & 1 & 1 & 0 & 0 & 1 \\
        $5 \times 10^5$ & 1 & 1 & 1 & 1 & 0.0015 & 0 & 0 & 0 \\
        $6 \times 10^5$ & 1 & 1 & 1 & 1 & 0.0978 & 0 & 0 & 0 \\
        $7 \times 10^5$ & 1 & 1 & 1 & 1 & 0 & 0 & 0 & 0 \\
        \hline
    \end{tabular}
    \caption{Posterior probabilities based on $H_0: F_X = F_Y$ from \mytest{} ($\varepsilon \approx 0.0657$) and from the PTtest for different sample sizes and values of the hyperparameter $c$. In all cases, one dataset was generated from a $N(0,1)$ and the other from a $t_{30}$.}
    \label{tab:n_t_compare}
\end{table}

Let $H_0: F_X = F_Y$, where $X$ represents data coming from the $N(0,1)$ and $Y$ from the $t_{30}$. \autoref{tab:n_t_compare} presents a comparison between \mytest{} and the PTtest \citep{holmes2015} in such context for multiple sample sizes. In order to highlight the difference between the methods while keeping them as similar as possible, we draw from the posterior of a Pólya tree process (PT, \cite{lavine1992, lavine1994}) for \mytest{} as well. We follow the recommendation of \cite{holmes2015} for choosing the hyperparameter $c$ of the PT and use $c \in \{1, 4, 7, 10\}$ for our comparisons. For both datasets, we apply a PT centered on $N(0,1)$.

Now, let us retrace all steps of the \hyperlink{blk:proc}{\mytest{} procedure}, but skipping the choice of $\alpha$ and step 4 altogether, since we are only interested in the posterior probabilities.
\begin{enumerate}
    \item The null hypothesis is $H_0: F_X = F_Y$. We use \eqref{eq:np_dist} as the dissimilarity function and follow the prior thresholding guideline presented in \autoref{sct:one_eps} for establishing $\varepsilon$. For each $c \in \{1, 4, 7, 10\}$, we obtain $\varepsilon$ such that $\mathbb{P}[Pg(H_0)] = 0.5$ and choose the most restrictive of them, which in this case resulted in $\varepsilon \approx 0.0657$.
    
    \item Since the number of parameters of the PT is infinite, we draw from a partially specified PT \citep{lavine1994} instead. Following \cite{hanson2002}, we set $\log_2 n \approx 20$ as the number of layers, $n$ being the largest sample size of \autoref{tab:n_t_compare}.
    
    \item From \autoref{thm:2sample_inf}, we conclude that, for any $(P_X, P_Y)$ obtained from the data,
    $$
    \inf_{(F_X,F_Y) \in H_0}d\left[(F_X,F_Y),(P_X,P_Y)\right] = d^*_C(P_X,P_Y).
    $$
    Now, let $\Omega$ be the sample space of both datasets and $(\Omega_i)_{i \in \{1, \cdots, I\}}$ be the sets obtained from the partition of the last layer of the PT. Then, if $F$ and $G$ come from partially specified PTs centered on the same distribution function,
    \begin{align*}
        \mathbb{P}\left(\frac{f(Z)}{g(Z)} > 1 \Big| Z \sim F \right) =&
        \sum_{i = 1}^I \mathbb{P}\left(\frac{f(Z)}{g(Z)} > 1 \Big|
        \{Z \sim F\} \cap \{Z \in \Omega_i\} \right)
        F\left(Z \in \Omega_i\right)\\
        =& \sum_{i = 1}^I \mathbb{P}\left(\frac{F\left(Z \in \Omega_i \right)}{G\left(Z \in \Omega_i \right)} > 1 \Big| \{Z \sim F\} \cap \{Z \in \Omega_i\} \right)
        F\left(Z \in \Omega_i\right)\\
        =& \sum_{i = 1}^I \mathbb{I}\left(\frac{F\left(Z \in \Omega_i \right)}{G\left(Z \in \Omega_i \right)} > 1 \right)
        F\left(Z \in \Omega_i\right),
    \end{align*}
    and thus \eqref{eq:np_dist} can be obtained analytically, easing the calculation of \eqref{eq:prob_pg}.
    
\end{enumerate}

From \autoref{tab:n_t_compare}, we see that the PTtest provides the desired outcome for smaller samples, but rejects when the sample size is large enough. Of course, rejecting the hypothesis is no fault of the PTtest since $H_0$ is false, but it is an indication that the test may be too rigorous on negligible differences that are perfectly compatible with real-world data when the sample size is large.

Unlike the PTtest, \mytest{} remains consistent for all cases as the sample size grows, and this is not to be confused with the method being permissive. Compared to the PTtest, its probability was generally lower for small sample sizes, but this is largely a consequence of choosing the more conservative $\varepsilon$. Moreover, the true dissimilarity between $N(0,1)$ and $t_{30}$ is around 0.005 and, when using this value for $\varepsilon$ instead, for no sample size did \mytest{} reach a probability other than 0.

\section{Application: Neuron spike analysis}
\label{sct:application}

In this section, we apply \mytest{} to data on the time between neuron spikes (in microseconds) of an epilepsy patient exposed to visual stimuli (pictures in varied contexts, each context represents an experiment). The first test evaluates if a Poisson process \citep{ross2009} can describe the data, while the second uses the median to verify if the neuron activity is similar across experiments. In both tests, we use a Dirichlet process (DP, \cite{ferguson1973}) with a centering distribution gamma and scaling parameter of 1. To stipulate the hyperparameters of the gamma distribution, we remove one of the experiments from the data and use its maximum likelihood estimates (MLE).

The original dataset \citep{faraut2018} is composed of 42 patients and the brain activity of their amygdala and hippocampus as they were subjected to the stimuli. The authors identified clusters of activity, which were assumed to represent individual neurons, and registered a total of 1576 individual neurons. We restricted the analysis to the neuron ``2494'' due to it having a high number of experiments applied (8 in total) and a reasonably high sample size in each experiment (minimum of 693, maximum of 2691). As for the experiments, we use the notation ``a-b'' to represent session b of experiment a, since the same type of visual stimuli might be presented at different times.

\begin{figure}[ht]
    \centering
    \includegraphics[width = \linewidth]{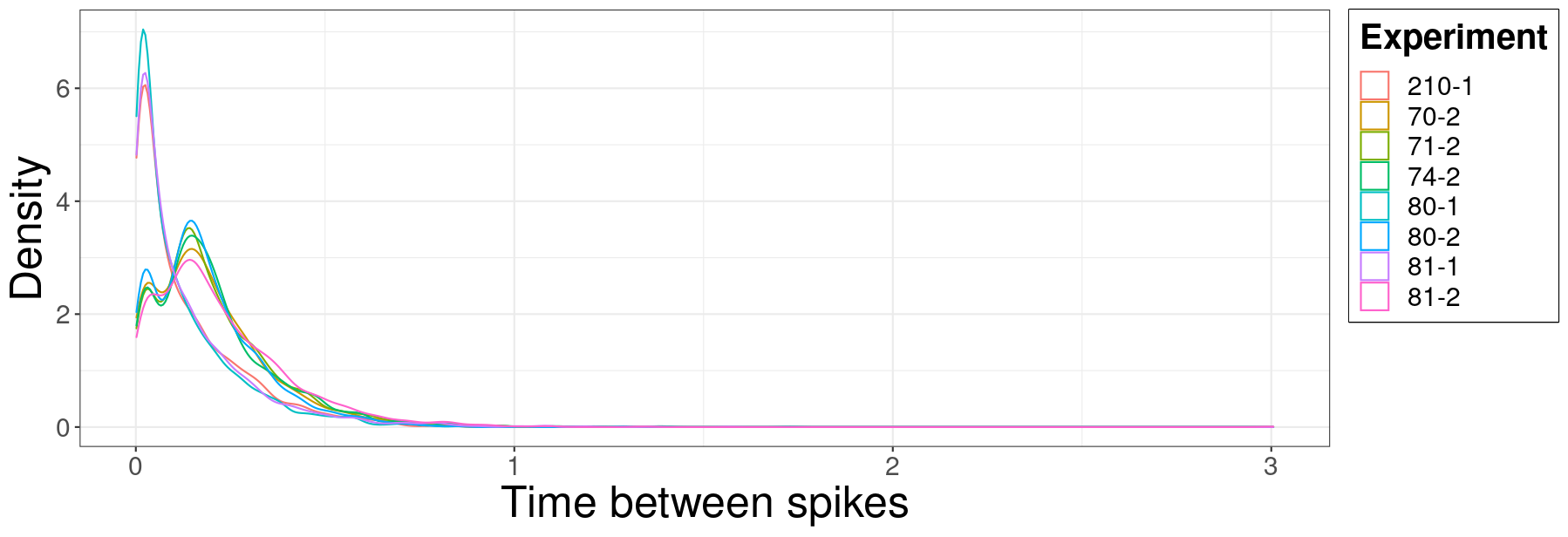}
    \caption{Smoothed sample density of the time between spikes of neuron ``2494'' for each experiment.}
    \label{fig:spike_den}
\end{figure}

\autoref{fig:spike_den} presents the smoothed sample densities for each experiment of neuron ``2494''. This plot alone already puts the assumption of a Poisson process into question, since some cases exhibit a bimodal behavior with peaks not that close to 0. As for the median, the densities of the experiments seem to be roughly divided in two groups, so the intragroup median might be similar enough.

For both tests, we use available information on how neurons work to set an upper bound for $\varepsilon$ through examples, a procedure described in \autoref{sct:many_eps}. Since a neuron spike typically lasts for 1 millisecond \citep[Section 1.1.1]{gerstner2002}, it would be physically impossible for another spike to be observed in such interval. This is also corroborated by the fact that the smallest time observed between spikes is 0.0016 second, i.e, 1.6 milliseconds. Therefore, if the difference between two distribution functions could be attributed to the 1 millisecond threshold, they should be deemed as practically equivalent.

We turn once again to the experiment excluded from the analysis to derive a distribution function of reference and to establish $\varepsilon$ from it. Let $A \sim Gamma(\hat{\alpha}, \hat{\beta})$, where $(\hat{\alpha}, \hat{\beta})$ are the MLE based on the removed experiment. If $tol = \pm 0.001$ represents the 1 millisecond threshold, we take $B \sim Gamma(\tilde{\alpha}, \tilde{\beta})$ such that $\mathbb{E}[B] = \mathbb{E}[A] + tol$ (the means differ by at most 1 millisecond) and $\mathbb{V}[B] = \mathbb{V}[A]$ (the variance remains the same). Then, we take $d(F_A, F_B)$ for both values of $tol$ and set the maximum as the proposal for $\varepsilon$.

\subsection{First test: Poisson process}
\label{sct:neuron_test1}

This case is a direct continuation of \autoref{ex:pois_process}, with $H_0: T_i|\lambda \stackrel{ind.}{\sim} Exp(1/\lambda), \lambda \in \mathbb{R}_{\ge 0}, i \in \{1, \cdots, n\}$, and $T_i$ representing the time-lapse between spikes. By using the $L^\infty$ distance from \autoref{eq:linf_pois} as the dissimilarity function and the strategy mentioned just above, we conclude that $\varepsilon \approx 0.0029$. Hence, we should expect a difference of at most 0.0029 between a distribution function drawn from the DP and the exponential distribution that is closest to it for any $x \in (0, \infty)$.

\autoref{fig:eps-pois} provides the largest $\varepsilon$ that leads to rejecting the hypothesis for each value of $\alpha$ in each experiment. From it, it is clear that taking $\varepsilon \approx 0.0029$ leads to rejection for all experiments, since $\mathbb{P}[\mathcal{P} \in Pg(H_0) | \boldsymbol{T}]$ becomes greater than 0 only when $\varepsilon \ge 0.06$. This result means that either the hypothesis should be rejected or that the choice of $\varepsilon$ was too strict. Considering that the values of $\varepsilon$ that would lead to non-rejection are considerably far from the initial estimate, we reject the hypothesis for all experiments.

\begin{figure}[ht]
    \centering
    \includegraphics[width = \linewidth]{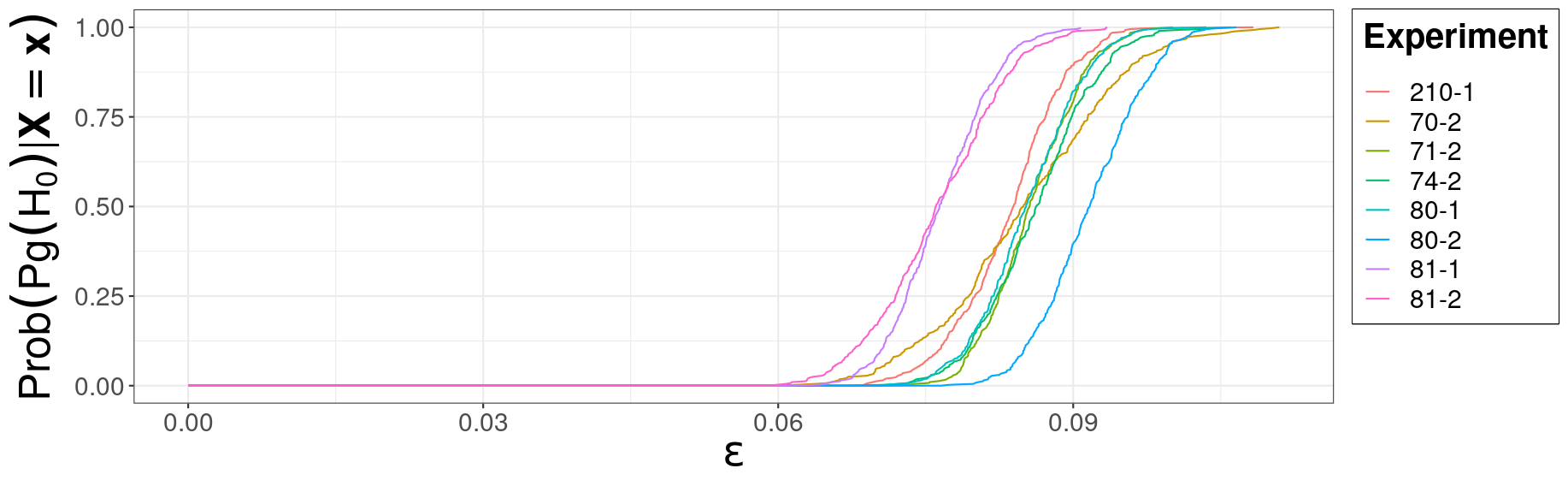}
    \caption{Largest $\varepsilon$ that entails rejection and the posterior probability of $Pg(H_0)$ for each experiment. The pragmatic hypothesis is expanded from $H_0: T_i|\lambda \stackrel{ind.}{\sim} Exp(1/\lambda), \lambda \in \mathbb{R}_{\ge 0}, i \in \{1, \cdots, n\}$.}
    \label{fig:eps-pois}
\end{figure}

\subsection{Second test: Median of time between spikes}
\label{sct:neuron_test2}

This second test is a particular case of the quantile test (\autoref{sct:quantile}, $p_0 = 0.5$), an instance of \mytest{} that has already been demonstrated in \autoref{ex:wage_gap}. In this case, the remaining steps required for performing \mytest{} are to assume a value for $x_0$ and to derive $\varepsilon$ for this case.
For the former, we use the experiment that was removed from the original data, which provides a sample median of around $0.1787$ second between spikes, implying that the null hypothesis can be expressed as
$$
    H_0: F(0.1787) = 0.5, \quad F \in \mathbb{F}.
$$
As for the latter, we once again turn to the scheme based on the 1 millisecond threshold, which when applied for the distance in \autoref{eq:l1_dist} results in $\varepsilon \approx 0.001$.

\begin{table}[ht]
    \centering
    \begin{tabular}{c|lll}
        \hline
        Experiment & Sample size & Sample median & $\alpha$ for rejecting $H_0$\\ 
        \hline
        70-2  & 693  & 0.1651 & 0.970\\
        71-2  & 2388 & 0.1668 & 1\\
        74-2  & 1834 & 0.1718 & 1\\
        80-1  & 2487 & 0.0693 & 0\\ 
        80-2  & 1919 & 0.1601 & 0.975\\ 
        81-1  & 2691 & 0.0785 & 0\\ 
        81-2  & 1547 & 0.1793 & 1\\
        210-1 & 2279 & 0.0795 & 0\\
        \hline
        \end{tabular}
    \caption{Comparison between experiments based on the sample median and the smallest value of $\alpha$ that would lead to the rejection of $H_0$ ($\varepsilon \approx 0.001$).}
    \label{tab:med_prob}
\end{table}

\autoref{tab:med_prob} presents the results of \mytest{} for this case, as well as information on the sample size and the sample median of each experiment. We observe that the test provides assertive decisions in all cases, requiring either a considerably high significance level $\alpha$ to reject or not requiring it at all. Following our intuition, the experiments whose sample medians are closer to 0.1787 are the ones that lead to non-rejection.

\begin{figure}[ht]
    \centering
    \includegraphics[width = \linewidth]{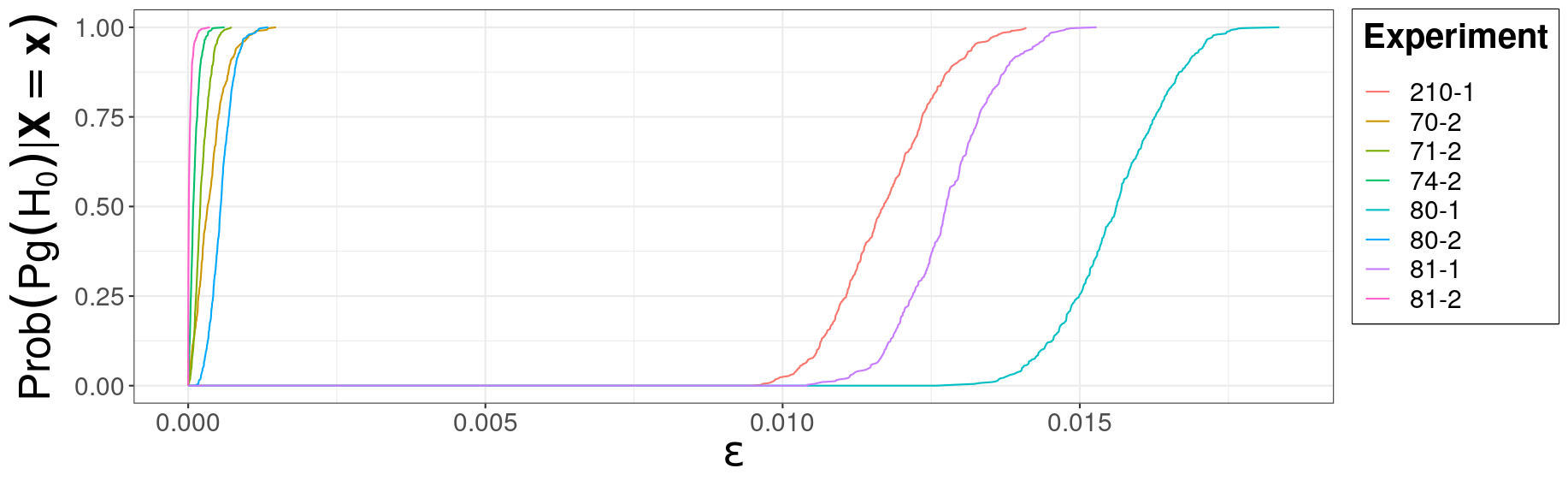}
    \caption{Largest $\varepsilon$ that entails rejection and the posterior probability of $Pg(H_0)$ for each experiment. The pragmatic hypothesis is expanded from $H_0: F(0.1787) = 0.5$, $F \in \mathbb{F}$.}
    \label{fig:eps-median}
\end{figure}

\autoref{fig:eps-median} provides more nuanced results, clearly contrasting between the experiments that were rejected and the ones that were not. While the conclusion of not rejecting the hypothesis for experiments whose curves reach their peak early is hardly contestable, the rejection for the other cases will depend on how strict is the choice of $\varepsilon$. Still, the clear divide between the curves is more evidence of the robustness of our decision.

\section{Discussion}
\label{sct:discussion}

\mytest{} offers a new paradigm for hypothesis testing, one that is theoretically sound, easy to apply and highly adaptable to practical settings. Moreover, although the four pragmatic versions covered here represent enhancements over nonparametric hypotheses routinely evaluated, there are still many other hypotheses left to be expanded. Cases that deal with multivariate or high-dimensional settings are probably the ones where greater attention should be directed, given how common these models have become.

The \hyperlink{blk:proc}{\mytest{} procedure} can be extended to the context of three-way testing -- which can accept, reject or remain undecided towards a hypothesis -- linking it more closely to the work of \cite{kruschke2018}. This can be done by switching \mytest{} for its three-way version, which also retains the monotonicity property (\autoref{prop:3way_mono}).
\begin{definition}[Three-way \mytest{}]
    \label{def:3way_test}
    Let $Pg(H_0, d, \varepsilon)$ be the pragmatic hypothesis and $\mathcal{P}$ be a random object over $\mathbb{H}$. The three-way \mytest{} is such that, for $0 \le \alpha_1 \le \alpha_2 \le 1$,
    \begin{itemize}
        \item If $\mathbb{P}\left(\mathcal{P} \in Pg(H_0)|\boldsymbol{X} = \boldsymbol{x}\right) \le \alpha_1$, reject the hypothesis;
        \item If $\alpha_1 < \mathbb{P}\left(\mathcal{P} \in Pg(H_0)|\boldsymbol{X} = \boldsymbol{x}\right) \le \alpha_2$, remain undecided;
        \item Otherwise, accept the hypothesis.
    \end{itemize}
\end{definition}

Other three-way testing procedures are the GFBST \citep{stern2017} and coherent agnostic tests in general \citep{esteves2016}. We note that all of these procedures heavily rely on pragmatic hypotheses, so the contributions in \autoref{sct:nph} can be of use even if one uses a procedure other than \mytest{}. Further still, if \mytest{} is adapted to evaluate a credibility region instead of $\mathbb{P}(\mathcal{P} \in Pg(H_0)|\boldsymbol{X} = \boldsymbol{x})$, it will acquire new properties that make it a fully coherent procedure in the sense presented by \cite{esteves2023}.

\begin{supplement}
\stitle{\texttt{protest} Package}
\sdescription{\texttt{R} package that implements \mytest{} and that reproduces some of the analyses in this paper. Its development version can be found in \href{https://github.com/rflassance/protest}{\faGithubSquare rflassance/protest}.}
\end{supplement}

\bibliographystyle{ba}
\bibliography{protest.bib}

\begin{acks}[Acknowledgments]
We thank Dani Gamerman, Julio M. Stern, Luben M. C. Cabezas and Luis G. Esteves for the fruitful conversations and suggestions regarding \mytest{}. This study was financed in part by the Coordenação de Aperfeiçoamento de Pessoal de Nível Superior - Brasil (CAPES) - Finance Code 001, FAPESP (grants 2019/11321-9 and 2023/07068-1) and CNPq (grants 309607/2020-5 and 422705/2021-7).
\end{acks}

\appendix
\section*{Appendix: Proofs}
\label{sct:proof}

\begin{proposition}[Monotonicity property of the three-way \mytest{}]
\label{prop:3way_mono}
Let $H_0^1, H_0^2 \subset \mathbb{H}$ be such that $Pg(H_0^1, d, \varepsilon_1) \supseteq Pg(H_0^2, d, \varepsilon_2)$ and $0 \le \alpha_1 \le \alpha_2 \le 1$. Then, the three-way \mytest{} (\autoref{def:3way_test}) has the monotonicity property, i.e.,
\begin{itemize}
    \item If the test rejects $Pg(H_0^1)$, then it rejects $Pg(H_0^2)$ as well;
    \item If the test remains undecided on $Pg(H_0^1)$, it does not accept $Pg(H_0^2)$.
\end{itemize}
\end{proposition}

\begin{proof}
    Let $\mathcal{P}$ be a random object on $\mathbb{H}$. Since $Pg(H_0^1) \supseteq Pg(H_0^2)$,
    \begin{equation}
        \label{eq:pg1_pg2}
        \mathbb{P}(\mathcal{P} \in Pg(H_0^1)) \ge \mathbb{P}(\mathcal{P} \in Pg(H_0^2)).
    \end{equation}
    If $Pg(H_0^1)$ is rejected,
    $$
        \mathbb{P}(\mathcal{P} \in Pg(H_0^1)) \le \alpha_1 \stackrel{\eqref{eq:pg1_pg2}}{\Longrightarrow} \mathbb{P}(\mathcal{P} \in Pg(H_0^2)) \le \alpha_1.
    $$
    If $Pg(H_0^1)$ remains undecided,
    $$
        \alpha_1 < \mathbb{P}(\mathcal{P} \in Pg(H_0^1)) \le \alpha_2 \stackrel{\eqref{eq:pg1_pg2}}{\Longrightarrow} \mathbb{P}(\mathcal{P} \in Pg(H_0^2)) \le \alpha_2.
    $$
\end{proof}

\begin{proof}[\textbf{Proof of \autoref{cor:monotonicity}}] The result follows by taking $\alpha_1 = \alpha_2$ in \autoref{prop:3way_mono}.
\end{proof}

\begin{theorem}[Infimum on a Hilbert space from a subspace of linear functionals]\label{thm:lin_inf}
    Let $\mathcal{H}$ be a Hilbert space and $\boldsymbol{b} = (b_1, b_2, \cdots, b_k)$ be a basis of linear functionals that constitutes the subspace $H \subset \mathcal{H}$. If $d(\cdot, \cdot)$ and $(\cdot, \cdot)$ are the distance function and the scalar product induced by the norm of $\mathcal{H}$ and $f_{\boldsymbol{\beta}} := \sum_{i = 1}^k \beta_i \times b_i$, $\boldsymbol{\beta} = (\beta_1, \beta_2, \cdots, \beta_k) \in \mathbb{R}^k$, then $\inf_{h \in H}d(h, g) = \inf_{\boldsymbol{\beta} \in \mathbb{R}^k}d(f_{\boldsymbol{\beta}}, g) = d(f_{\hat{\boldsymbol{\beta}}}, g)$ for $g \in \mathcal{H}$, where
    $$
        \hat{\boldsymbol{\beta}} = A_{\boldsymbol{b}}^{-1} \times \boldsymbol{g}_{\boldsymbol{b}},
        \quad A_{\boldsymbol{b}} =
        \left(\begin{array}{cccc}
             (b_1, b_1) & (b_2, b_1) & \cdots & (b_k, b_1) \\
             (b_1, b_2) & (b_2, b_2) & \cdots & (b_k, b_2) \\
             \vdots     & \vdots     & \ddots & \vdots     \\
             (b_1, b_k) & (b_2, b_k) & \cdots & (b_k, b_k)
        \end{array}\right),
        \quad \boldsymbol{g}_{\boldsymbol{b}} =
        \left(\begin{array}{c}
             (g, b_1)\\
             (g, b_2)\\
             \vdots \\
             (g, b_k)\\
        \end{array}\right).
    $$
\end{theorem}
\begin{proof}[\textbf{Proof of \autoref{thm:lin_inf}}]
    By construction, $H$ is a closed linear subspace. From corollary 5.4 of \cite{brezis2011}, for each $g \in \mathcal{H}$, $f_{\hat{\boldsymbol{\beta}}}$ is characterized by
    $$
        (g - f_{\hat{\boldsymbol{\beta}}}, f_{\boldsymbol{\beta}}) = \sum_{j = 1}^{k} \beta_j (g - f_{\hat{\boldsymbol{\beta}}}, b_j) = 0, \quad \forall \boldsymbol{\beta} \in \mathbb{R}^k \Longrightarrow (g - f_{\hat{\boldsymbol{\beta}}}, b_j) = 0, \quad \forall j \in \{1, 2, \cdots, k\}.
    $$
    Therefore,
    $$
        (g - f_{\hat{\boldsymbol{\beta}}}, b_j) = (g, b_j) - \sum_{i = 1}^{k}\hat{\beta}_i(b_i, b_j) = 0, \quad \forall j \in \{1, 2, \cdots, k\},
    $$
    thus leading to the linear system
    $$
        \left\{\begin{array}{c}
             \sum_{i = 1}^{k}\hat{\boldsymbol{\beta}}_i(b_i, b_1) = (g, b_1) \\
             \sum_{i = 1}^{k}\hat{\boldsymbol{\beta}}_i(b_i, b_2) = (g, b_2) \\
             \vdots \\
             \sum_{i = 1}^{k}\hat{\boldsymbol{\beta}}_i(b_i, b_k) = (g, b_k)
              
        \end{array}\right. \Longrightarrow A_{\boldsymbol{b}} \times \hat{\boldsymbol{\beta}} = \boldsymbol{g}_{\boldsymbol{b}} \Longrightarrow \hat{\boldsymbol{\beta}} = A_{\boldsymbol{b}}^{-1} \times \boldsymbol{g}_{\boldsymbol{b}}.
    $$
\end{proof}

\begin{proof}[\textbf{Proof of \autoref{thm:lm_inf}}]
    We note that $\mathbb{H} \equiv L^2(\mathcal{X}, \sigma(\mathcal{X}), \mathbb{P})$ is a Hilbert space and that $\spn\{b_1, b_2, \cdots, b_k\} = H_0$, therefore \autoref{thm:lin_inf} follows by switching $H$ for $H_0$. Moreover,
    \begin{align*}
        (b_i, b_j) &= \int_{\mathcal{X}} b_i(\boldsymbol{x})b_j(\boldsymbol{x})d\mathbb{P}(\boldsymbol{x}) = \mathbb{E}[b_i(\boldsymbol{X})b_j(\boldsymbol{X})], \quad \forall i,j \in \{1, 2, \cdots, k\};\\
        (g, b_i) &= \int_{\mathcal{X}} g(\boldsymbol{x})b_i(\boldsymbol{x})d\mathbb{P}(\boldsymbol{x}) = \mathbb{E}[g(\boldsymbol{X})b_i(\boldsymbol{X})], \quad \forall i \in \{1, 2, \cdots, k\}.
    \end{align*}
\end{proof}

\begin{proof}[\textbf{Proof of \autoref{thm:quant_inf}}]
The proof is done in parts.
\begin{itemize}
    \item If $P(x_0) = p_0$, then $\inf_{P_0 \in H_0} d(P_0, P) = \int_{a}^b |p_0 - P(x)|dx = \int_{x_0}^{x_0} |p_0 - P(x)|dx = 0$.
    \begin{subproof}
        If $P(x_0) = p_0$, then $P \in H_0$. If that is the case,
        \begin{equation*}
            \inf_{P_0 \in H_0} \int_{-\infty}^\infty |P_0(x) - P(x)|dx = \int_{-\infty}^\infty |P(x) - P(x)|dx = \int_{-\infty}^\infty 0 dx = 0.
        \end{equation*}
    \end{subproof}
    
    \item If $P(x_0) < p_0$, then $\inf_{P_0 \in H_0} d(P_0, P) = \int_{a}^{b} |p_0 - P(x)|dx$.
    \begin{subproof}
        $P(x_0) < p_0 \Longrightarrow a = x_0 \text{ and } b = P^{-1}(p_0)$. Let $P^*(\cdot)$ be such that
        $$
        P^*(x) := \left\{
        \begin{array}{ll}
            p_0, & \text{if } x \in [x_0, b];\\
            P(x), & \text{otherwise}.
        \end{array}
        \right.
        $$
        Thus, proving the result is equivalent to showing that
        $$
            \inf_{P_0 \in H_0} d(P_0, P) = d(P^*, P) = \int_{a}^{b} |p_0 - P(x)|dx.
        $$
        
        Suppose by contradiction that $\exists P' \in H_0: d(P^*, P) > d(P', P)$. Hence,
        \begin{equation}\label{eq:abs1}
            \int_{a}^{b} |p_0 - P(x)|dx > \int_{-\infty}^{\infty} |P'(x) - P(x)|dx \ge \int_{a}^{b} |P'(x) - P(x)|dx.
        \end{equation}
        For $x \in [a,b]$, $P(x) \le p_0 \le P'(x) \Longrightarrow P(x) - p_0 \le 0 \le P'(x) - p_0$. Since $[P'(x) - p_0] - [P(x) - p_0] = |P'(x) - p_0| + |P(x) - p_0|$,
        \begin{align*}
            \int_{a}^{b} |P'(x) - P(x)|dx =& \int_{a}^{b} |[P'(x) - p_0] - [P(x) - p_0]|dx\\
            =& \int_{a}^{b} |P'(x) - p_0| + |P(x) - p_0|dx \ge \int_{a}^{b}|p_0 - P(x)|dx,
        \end{align*}
        which contradicts \eqref{eq:abs1}, therefore $\inf_{P_0 \in H_0} d(P_0, P) = \int_{a}^{b} |p_0 - P(x)|dx$.
    \end{subproof}
    
    \item If $P(x_0) > p_0$, then $\inf_{P_0 \in H_0} d(P_0, P) = \int_{a}^{b} |p_0 - P(x)|dx$.
    
    \begin{subproof}
        $P(x_0) > p_0 \Longrightarrow b = x_0$. Let $(P_n^*)_{n \ge 1}$ be a sequence of distribution functions such that
        $$
        P_n^*(x) := \left\{
        \begin{array}{ll}
            p_0, & \text{if } x \in \left[a, x_0 + \frac{1}{n}\right);\\
            P(x), & \text{otherwise}.
        \end{array}
        \right.
        $$
        By construction, $P_n^* \in H_0, \forall n \ge 1$, and
        $$
            d(P_n^*, P) = \int_{a}^{x_0 + \frac{1}{n}} |p_0 - P(x)|dx = \int_{a}^{x_0} |p_0 - P(x)|dx + \int_{x_0}^{x_0 + \frac{1}{n}} |p_0 - P(x)|dx,
        $$
        which converges decreasingly to $\int_{a}^{x_0} |p_0 - P(x)|dx$ as $n \rightarrow \infty$.
        
        The proof follows by contradiction. Suppose $\exists P' \in H_0 : \inf_{P_0 \in H_0} d(P_0, P) = d(P', P) \neq \int_{a}^{b} |p_0 - P(x)|dx$. Similarly to the previous subproof,
        \begin{align*}
            \int_{a}^{b} |P'(x) - P(x)|dx = & \int_{a}^{b} |[P(x) - p_0] - [P'(x) - p_0]|dx\\
            = & \int_{a}^{b}|P(x) - p_0|dx + \int_{a}^{b} |P'(x) - p_0|dx\\
            \ge & \int_{a}^{b}|p_0 - P(x)|dx,
        \end{align*}
        and thus $d(P', P) > \int_{a}^{b}|p_0 - P(x)|dx$. But since $d(P_n^*, P) \overset{n \rightarrow \infty}{\longrightarrow} \int_{a}^{b}|p_0 - P(x)|dx$, then $\exists n_0 \in \mathbb{N}$ such that, $\forall n \ge n_0$,
        \begin{align*}
            &d(P_n^*, P) < d(P', P) \Longrightarrow \inf_{P_0 \in H_0} d(P_0, P) \neq d(P', P),
        \end{align*}
        therefore $\inf_{P_0 \in H_0} d(P_0, P) = \int_{a}^{b} |p_0 - P(x)|dx$.
    \end{subproof}
    Based on each of the subproofs presented, we can safely conclude that
    $$\inf_{P_0 \in H_0} d(P_0, P) = \int_{a}^{b} |p_0 - P(x)|dx.$$
\end{itemize}
\end{proof}

\begin{proof}[\textbf{Proof of \autoref{thm:2sample_inf}}]
Without loss of generality, we assume that $\mathbb{H} = \mathbb{F}_X \times \mathbb{F}_Y = \mathbb{F}_X \times \mathbb{F}_X = \mathbb{F}_Y \times \mathbb{F}_Y$. After all, if $X$ and $Y$ were defined on different distribution spaces, we could simply take $\mathbb{F} := \mathbb{F}_X \cup \mathbb{F}_Y$ and use this space instead.

The null hypothesis asserts that, as long as $F_X = F_Y$, the distribution function of both random variables can be any element of $\mathbb{F}$. Thus, if $\Omega$ is the sample space, 
\begin{equation*}
    H_0: (F_X, F_Y) \in \mathbb{F} \times \mathbb{F} : F_X(z) = F_Y(z), \forall z \in \Omega.
\end{equation*}
Therefore, 
\begin{align*}
Pg(H_0) =& \left\{(P_X, P_Y) \in \mathbb{F} \times \mathbb{F}: \inf_{(F_X, F_Y) \in H_0} d[(F_X, F_Y),(P_X, P_Y)] <\varepsilon \right\}\\
=& \left\{(P_X, P_Y) \in \mathbb{F} \times \mathbb{F}: \inf_{P_0 \in \mathbb{F}} d[(P_0, P_0),(P_X, P_Y)] <\varepsilon \right\}.
\label{eq:2sample_general}
\end{align*}
From \eqref{eq:2sample_cond},
\begin{equation}
    \label{eq:2sample_l1}
    \inf_{P_0 \in \mathbb{F}} d[(P_0, P_0),(P_X, P_Y)] = \inf_{P_0 \in \mathbb{F}}[d^*(P_0,P_X) + d^*(P_0,P_Y)].
\end{equation}
Now, since $d^*$ is a distance function, the properties of symmetry and triangle inequality \citep{kreyszig1978} imply that
\begin{equation}
\inf_{P_0 \in \mathbb{F}}[d^*(P_0,P_X) + d^*(P_0,P_Y)] = \inf_{P_0 \in \mathbb{F}}[d^*(P_X, P_0) + d^*(P_0,P_Y)] \ge d^*(P_X,P_Y).
\label{eq:2sample_final}
\end{equation}
Since $P_X \in \mathbb{F}$, the equality in \eqref{eq:2sample_final} is guaranteed if $P_0 = P_X$.
\end{proof}

\end{document}